\title{Statistical Analysis of Multi-Antenna Relay Systems\\ and Power Allocation Algorithms in a Relay with Partial Channel State Information}

\documentclass[journal,doublecolumn,10pt,twoside]{IEEEtran}

\normalsize
\usepackage{lettrine}
\usepackage[stable]{footmisc}
\usepackage{soul}
\usepackage{amsthm}
\usepackage{amsmath}
\usepackage{amssymb}
\usepackage{graphicx,subfigure}
\usepackage[countmax]{subfloat}
\usepackage{color}
\usepackage{array}
\usepackage{url}
\usepackage{cite,balance}
\usepackage{setspace}
\usepackage{siunitx}
\DeclareGraphicsExtensions{{.eps}{.pdf}{.jpeg}}
\newtheorem{thm}{Theorem}

\begin{document}

\author{Mehdi~M.~Molu, Alister Burr and Norbert~Goertz
  \thanks{Mehdi~M.~Molu and Alister Burr are with the
    Department of Electronics, University of York,
    UK, email:\{mehdi.molu,alister.burr\}@york.ac.uk. Norbert Goertz is with the Institute of Telecommunications,
    Vienna University of Technology, Austria, email: norbert.goertz@nt.tuwien.ac.at.

The paper is accepted for presentation in part at the European Conference on Networks and Communications~(EuCNC), 2015, Paris, France}  }
\markboth{IEEE Transactions on Wireless Communications, Accepted for Publication}{}
\maketitle

\begin{abstract}
The performance of a dual-hop MIMO relay network is studied in this paper. The relay is assumed to have access to the statistical channel state information of its preceding and following channels and it is assumed that fading at the antennas of the relay is correlated. The cumulative density function~(cdf) of the received SNR at the destination is first studied and closed-form expressions are derived for the asymptotic cases of the fully-correlated and non-correlated scenarios; moreover, the statistical characteristics of the SNR are further studied and an approximate cdf of the SNR is derived for arbitrary correlation. The cdf is a multipartite function which does not easily lend itself to further mathematical calculations, e.g., rate  optimization. However, we use it to propose a simple power allocation algorithm which 
we call ``proportional power allocation". The algorithm is explained in detail for the case of two antennas and three antennas at the relay and the extension of the algorithm to a relay with an arbitrary number of the antennas is discussed. 
Although the proposed method is not claimed to be optimal, 
the result is indistinguishable from the benchmark obtained using exhaustive search. The simplicity of the algorithm combined with its precision is indeed attractive from the practical point of view. 
\end{abstract}

\section{Introduction}
\label{Sec:Introduction}
  \lettrine{I}t was proved by van der Meulen in \cite{Va:1971} that the capacity of a three node communication system can, potentially, be larger than the capacity of a point-to-point communication system. Consequently, 
analysis of communication systems in which the transceiver nodes cooperatively transmit their data to an intended final receiver has been a rather active field of research
 in last decade and numerous  papers including, e.g., \cite{LaTsWo:2004,SeErAa-Part1:2003,SeErAa-Part2:2003,CoGa:1979} have investigated cooperative communication systems. 
During the infancy of the concept of cooperative communications, substantial work was carried out, investigating  cooperating nodes with single antennas. Several promising relaying protocols were proposed; among them, Amplify and Forward~(AF) is  intensively studied in the literature; hence, in this paper, we will focus on AF relaying too.

Employing multiple antennas in communication nodes is another technique proved to be capable of enhancing transmission rates (see e.g., \cite{Te:1995}). Employing multiple antennas in the nodes of a cooperative communication system has been an active research trend during the last few years.
 Assuming multiple antennas at the relay, one major task is to design a suitable amplification matrix in the relay. 
Indeed, depending on the
available Channel State Information (CSI) at the relay, the amplification matrix can, potentially, be different.
Moreover different communication systems can demand the optimization of different  desired performance measure; hence, different ``optimal''
relaying protocols will exist: for instance, the non-regenerative
relaying matrices, e.g., in \cite{GuHa:2008,MoCh:2009,KhRo:2010,cumanan2013mmse}, are
designed to minimize Mean Square Error (MSE) but other relaying
matrices,
e.g., in \cite{MoGo:2014TWC,MuViAg:2007,TaHu:2007,DhMcMaBe:2011,JeSeLeKiKi:2012},
are assumed to maximize the achievable rates. 

Assuming statistical CSI at a transceiver node is interesting from a practical point of view; in particular, in rapidly changing channels, assuming perfect CSI in a relay node is, indeed, unrealistic, hence, a large body of the literature investigates AF cooperative systems wherein a single antenna relay node has access  only to the statistical CSI~(see, e.g.,\cite{FaBe:2008TWC}). Note that single antenna AF relaying systems, with statistical CSI at the relay, are usually referred to as ``fixed gain" AF relaying. In spite of the importance of cooperative communication systems with statistical CSI knowledge, very few papers consider the problem when the relay node is equipped with multiple antennas. Moreover, except \cite{MoGo:2014TWC}, we are not aware of any other paper assuming fading correlation in the relay when only the covariance of the channels is known to the relay. Note that fading correlation at a transceiver can be due to an unobstructed node or space limits at the node which forces the  antennas to be closely located. Justifications to assume transceivers with fading correlation can be found in \cite{Shiu:1999-KAP,Shiu:2000-TC}.

The first contribution of this paper is to provide a statistical analysis of the received Signal to Noise Ratio~(SNR) at the destination. There are two major motivations for studying the statistical characteristics of the SNR: 
\begin{itemize}
\item Outage probability is directly related to received SNR. Indeed, the cumulative density function~(cdf) of the SNR corresponds to the outage probability, and so, the cdf of SNR will be derived in this paper.
\item By deriving the cdf of SNR, the mathematical complexity of direct maximization of the achievable rate~(i.e., ``optimal'' power allocation) will be revealed. It will be an excellent motivation for devising alternative approaches with reasonable complexity.
\end{itemize} 

Accordingly, the second major contribution of the paper is to study the problem of power allocation in the relay  and hence to devise a new and simple power allocation scheme for multi-antenna relays. 
 \cite{JeSeLeKiKi:2012,DhMcMaBe:2011,Zappone:2014-TSP,MoGo:2014TWC} consider the similar problem of the power allocation in the relay when statistical CSI is available  in the nodes~(either the source or relay nodes). However, while  \cite{JeSeLeKiKi:2012} considers the high SNR regime of the system, \cite{DhMcMaBe:2011,Zappone:2014-TSP} assumes correlation at the source node. In \cite{MoGo:2014TWC}, we study a cooperative communication system  wherein the relay node is equipped with multiple antennas that are spatially correlated. The considered system is studied only at low SNR and it is proved that Largest Eigenmode Relaying~(LER) is the optimal transmission method at low SNR; however, the system is not studied in the moderate and high SNR region. To the best of our knowledge, the design of an amplification matrix in an AF cooperative system where \emph{only} the statistical CSI is known to the relay is an open problem, and one that will be tackled in this paper.  We provide a scheme which operates in the regime beyond that where LER is optimal, and whose performance is indistinguishable from the benchmark provided by exhaustive search.

This paper is organized as follows: In Section~\ref{Sec:System Model}, the system model is introduced and some preliminary existing results are recalled. Section~\ref{Sec:Distribution of Received SNR at Destination} deals with characterizing the statistics of the SNR at the destination. In Section~\ref{Sec:Two Antennas in the Relay}, a simple power allocation algorithm is introduced for a relay with only two antennas; the proposed algorithm is called ``proportional power allocation" and has been extended for a system with multiple antenna relay node in section~\ref{Sec:Three Antennas in the Relay} and \ref{Sec:n Antennas in the Relay} and, finally, the results are summarized in Section~\ref{Sec:Conclusion}.

\section{System Model and Preliminaries}
\label{Sec:System Model}
\subsection{Notation}
\label{Subsec:Notation}
Matrices are represented by boldface upper cases
(${\boldsymbol{H}}$). Column and row vectors are denoted by boldface
lower cases (${\boldsymbol{h}}$), and ${h_i}$ indicates the ${i}$-th element
of ${\boldsymbol{h}}$. The superscript ${(\cdot)^H}$ stands for Hermitian
transposition. We refer to the identity matrix by
${\boldsymbol{I}}$. The expectation operation is indicated by
${\mathbb{E}\lbrace\cdot\rbrace}$, the probability of a random variable is indicated by $\mathbb{P}(\cdot)$ and ${f_X(x)}$ is reserved for probability density
functions (pdf) of random variable ${X}$; ${\boldsymbol{\Lambda}_\Sigma}$ represents a diagonal
matrix with elements organized in descending order and
${\lambda_i^{\Sigma}}$ denotes the ${i}$-th diagonal element of
${\boldsymbol{\Lambda}_\Sigma}$. For simplicity of notation,
${(\lambda_i^{\Sigma})^2}$ is abbreviated by ${\lambda_i^{\Sigma2}}$. The
trace of a matrix is denoted by ${\mathrm{Tr}(\cdot)}$.

\subsection{System Model}
\label{Subsec:System Model}
In this paper a dual hop, half duplex MIMO communication system is investigated. 
Assume a source node (equipped with ${n_\text{S}}$ antennas)  transmits data to a single antenna  destination  via an intermediate relay node which has ${n_\text{R}}$ antennas. The proposed system models the downlink of a wide range of communication systems in which the user terminal is equipped with single antenna due to space limitation, for instance, cellular networks or sensor networks. Moreover, \emph{fixed-gain} AF cooperative systems with multiple antennas at the relay is an open problem which has received little attention and so the proposed system model is a good step forward for understanding fixed gain AF systems. It is assumed that a direct link
between the source and the destination is not available. The half
duplex constraint is accomplished by time sharing between the source
and the relay; i.e.~each transmission period is divided into two time
slots: the source transmits during the first time slot and the relay
during the second one. The relay remains silent during the source
transmission and vice versa. It is assumed that the source does not
have access to any statistical or instantaneous channel state
information (CSI). Moreover, it is assumed that the antennas in the source node are sufficiently far apart and so no correlation is assumed at the source . The signal received at the relay
(${\boldsymbol{y}_\text{R}}$) due to the source transmission
is given by 
\begin{equation}
\label{eq:S-RTransmission}
\boldsymbol{y}_\text{R}=\boldsymbol{H}_1 \boldsymbol{x}+\boldsymbol{w}_\text{R}
\end{equation}
where the ${n_\text{R}\times n_\text{S}}$ matrix ${\boldsymbol{H}_1}$
represents the channel between the source and the relay. With ${P_\text{S}}$ the power
constraint of the source, the column vector ${\boldsymbol{x}}$ is the
signal transmitted from the source with
${\boldsymbol{Q}=\mathbb{E}(\boldsymbol{x}\boldsymbol{x}^H)=\frac{P_\text{S}}{n_\text{S}}\boldsymbol{I}_{n_\text{S}}}$
and the column vector ${\boldsymbol{w}_\text{R}}$ represents the
receiver noise in the relay with elements independently drawn from a
complex Gaussian random variable with variance ${N_0}$. In this paper, it is assumed that spatial correlation occurs at the relay; the correlation can be due to space limit in the relay or due to fading correlation due to unobstructed relay node.  ${\boldsymbol{\Sigma}}$ represents the correlation matrix at the 
relay and therefore, using the Kronecker  model, ${\boldsymbol{H}_1}$ can be written as
\begin{equation}
\label{eq:Kronecker-H1}
\boldsymbol{H}_1 = \boldsymbol{\Sigma}^\frac{1}{2}\boldsymbol{H}_{1w}
\end{equation}
where  elements of ${\boldsymbol{H}_{1w}}$ are i.i.d., zero
mean, unit variance complex Gaussian random variables, independent of
each other.
The relay
multiplies ${\boldsymbol{y}_\text{R}}$ by the gain matrix ${\boldsymbol{F}}$
and forwards it to the destination. Then, the received signal at the
destination is 
\begin{eqnarray}
\label{eq:R-DTransmission}
y_\text{D}&=&\boldsymbol{h}_2 \boldsymbol{F}\boldsymbol{y}_\text{R}+w_\text{D}
\\
&=&\boldsymbol{h}_2\boldsymbol{F}\boldsymbol{H}_1 \boldsymbol{x}+\boldsymbol{h}_2\boldsymbol{F}
\boldsymbol{w}_\text{R}+w_\text{D}\nonumber
\end{eqnarray}
where the row vector ${\boldsymbol{h}_2}$ indicates the channel between
the relay and the destination; ${w_\text{D}}$ represents
the receiver noise at the destination. For simplicity, we assume that
${w_\text{D}}$ and ${\boldsymbol{w}_\text{R}}$ are statistically independent and identical, i.e.~${N_{0,w_\text{D}}=N_{0,\boldsymbol{w}_\text{R}}= N_0}$.
\\Due to the spatial correlation ${\boldsymbol{\Sigma}}$ at the relay, one  can factorize $\boldsymbol{h}_2$ as 
\begin{eqnarray}
\label{eq:Kronecker-h2}
\boldsymbol{h}_2 = \boldsymbol{h}_{2w} \boldsymbol{\Sigma}^{\frac{1}{2}}
\end{eqnarray} 
where  elements of ${\boldsymbol{h}_{2w}}$ are i.i.d., zero
mean, unit variance complex Gaussian random variables, independent of each other. Justification to assume
transceivers with spatial correlation can be found
in \cite{ShFoGaKa:2000, Sh:1999}. The correlation matrix ${\boldsymbol{\Sigma}}$ in the relay is decomposed using spectral decomposition as 
\begin{eqnarray}
\label{eq:Sigma-Decomposition}
\boldsymbol{\Sigma} = \boldsymbol{U}_\Sigma\boldsymbol{\Lambda}_\Sigma\boldsymbol{U}^H_\Sigma
\end{eqnarray}
where ${\boldsymbol{U}_\Sigma}$ is a unitary matrix whose columns are 
the eigenvectors corresponding to ${\boldsymbol{\Sigma}}$, and
${\boldsymbol{\Lambda}_\Sigma}$ is a diagonal matrix with the eigenvalues of
${\boldsymbol{\Sigma}}$ in decreasing order, i.e. ${\boldsymbol{\Lambda}_\Sigma = \mathrm{diag}[\lambda^{\Sigma}_1,\lambda^{\Sigma}_2, \cdots, \lambda^{\Sigma}_{n_\text{R}}]}$ where ${\lambda^{\Sigma}_1\ge \lambda^{\Sigma}_2\ge \cdots \ge \lambda^{\Sigma}_{n_\text{R}}\ge 0}$. Moreover, some of ${\lambda^{\Sigma}_i}$s \emph{can} possibly be zero.

\subsection{Preliminaries\footnote{The results in this subsection are taken from \cite{MoGo:2014TWC}. Interested reader is recommended to read \cite{MoGo:2014TWC} for a complete proof of the the ideas and the derivations; however, in order to make the paper self-contained and also for consistency of the notation, the relevant results from \cite{MoGo:2014TWC} are provided in this subsection. }}
\label{Sec:Preliminaries}
Ergodic capacity is one of the main performance criterion investigated in this paper. Using (\ref{eq:R-DTransmission}), the ergodic capacity of the
system is defined as 
\begin{eqnarray}
\label{eq:ErgodicCapactityDefinition}
C_{av}= \frac{1}{2} \max_{\substack{\boldsymbol{Q}=\frac{P_\text{S}}{n_\text{S}}\boldsymbol{I} \\F: \mathbb{E}\lbrace\| \boldsymbol{F} \boldsymbol{y}_\text{R} \|^2\rbrace\le P_\text{R}}} \mathbb{E} \lbrace C(\boldsymbol{H}_1,\boldsymbol{h}_2,\boldsymbol{F})\rbrace
\end{eqnarray}
 where ${C(\boldsymbol{H}_1,\boldsymbol{h}_2,\boldsymbol{F})}$ is the conditional transmission rate. For simplicity of notation, ${C(\boldsymbol{H}_1,\boldsymbol{h}_2,\boldsymbol{F})}$ is abbreviated by ${C(\cdot)}$ in the rest of the paper. Assuming perfect CSI of ${\boldsymbol{H}_1}$ and ${\boldsymbol{H}_2}$ at the destination and
${\boldsymbol{Q}=\frac{P_\text{S}}{n_\text{S}}\boldsymbol{I}_{n_\text{S}}}$
(equal transmit power from each antenna in the source,
because no channel knowledge is available there), the conditional mutual information ${C(\cdot)}$ for given channel matrices is 
\begin{eqnarray}
\label{eq:ErgodicCapactityLog}
C(\cdot) = \log\big(1+\frac{P_\text{S}}{n_\text{S}}\frac{ \boldsymbol{h}_2\boldsymbol{F}\boldsymbol{H}_1\boldsymbol{H}^H_1\boldsymbol{F}^H\boldsymbol{h}^H_2}{N_0(1+\boldsymbol{h}_2\boldsymbol{F}\boldsymbol{F}^H\boldsymbol{h}^H_2)}\big)
\end{eqnarray}
where
${N_0(1+\boldsymbol{h}_2\boldsymbol{F}\boldsymbol{F}^H\boldsymbol{h}^H_2)}$
is the total equivalent noise power which is assumed to remain constant for 
coherence time: we make a block fading assumption. In \cite{MoGo:2014TWC}, 
${\boldsymbol{F}}$ is found to be symmetric as 
\begin{eqnarray}
\label{eq:Optimal-G}
\boldsymbol{F} =\boldsymbol{G}^{\frac{1}{2}}
\end{eqnarray}
where the gain matrix ${\boldsymbol{G}}$ is derived as 
\begin{eqnarray}
\label{eq:G-Decomposition}
 \boldsymbol{G} = \boldsymbol{U}_\Sigma \boldsymbol{\Lambda}_G \boldsymbol{U}^{H}_{\Sigma}
\end{eqnarray} 
where $\boldsymbol{U}_\Sigma$ is the unitary matrix defined in \eqref{eq:Sigma-Decomposition}  and   ${\boldsymbol{\Lambda}_G = \mathrm{diag}[\lambda^{G}_1,\lambda^{G}_2, \cdots, \lambda^{G}_{n_\text{R}}]}$. Note that ${\lambda^G_i}$ values are to be specified according to the power constraint of the relay so that the maximization in (\ref{eq:ErgodicCapactityDefinition}) is accomplished; indeed, this is one of the main tasks to be handled in this paper.

Assuming (\ref{eq:Kronecker-h2}), (\ref{eq:Sigma-Decomposition}), (\ref{eq:Optimal-G}) and (\ref{eq:G-Decomposition}), the power constraint in the relay (i.e.~${\mathbb{E}\lbrace\| \boldsymbol{G}^\frac{1}{2} \boldsymbol{y}_\text{R} \|^2\rbrace\le P_\text{R}}$ in (\ref{eq:ErgodicCapactityDefinition})) is 
{\begin{eqnarray}
\label{eq:RelayPowConstraintLambda}
P_\text{S}\mathrm{Tr}(\boldsymbol{\Lambda}_\Sigma\boldsymbol{\Lambda}_{G}) + N_0\mathrm{Tr}(\boldsymbol{\Lambda}_{G})= P_\text{R}
\end{eqnarray}}
Note that the capacity will be achieved by consuming the entire power at the relay, and so, we assume --equality-- in~\eqref{eq:RelayPowConstraintLambda} instead of --inequality--.  
By combining (\ref{eq:Kronecker-H1}), (\ref{eq:Kronecker-h2}), (\ref{eq:ErgodicCapactityLog}), (\ref{eq:Optimal-G}) and assuming ${\gamma_\text{S}= P_\text{S}/N_0}$, one can write (\ref{eq:ErgodicCapactityLog}) as follows 
\begin{eqnarray}
\label{eq:ErgodicCapactityLog-gammaD}
C(\cdot) = \log(1 + \underbrace{\frac{\gamma_\text{S} \boldsymbol{h}_{2w}\boldsymbol{\Lambda}_{\Sigma}\boldsymbol{\Lambda}_G^\frac{1}{2}\boldsymbol{H}_{1w}\boldsymbol{H}^H_{1w}\boldsymbol{\Lambda}_{\Sigma}\boldsymbol{\Lambda}_G^\frac{1}{2}\boldsymbol{h}^H_{2w}}{n_\text{S}(1+\boldsymbol{h}_{2w}\boldsymbol{\Lambda}_\Sigma\boldsymbol{\Lambda}_G\boldsymbol{h}^H_{2w})}}_{\gamma_\text{D}})
\end{eqnarray}
where ${\gamma_{\text{D}}}$ represents received SNR in the destination which is a function of ${\boldsymbol{H}_{1w}}$, ${\boldsymbol{h}_{2w}}$, the correlation eigenvalues matrix ${\boldsymbol{\Lambda}_{\Sigma}}$ and the eigenvalues of the $\boldsymbol{G}$ matrix, i.e.  ${\boldsymbol{\Lambda}_{G}}$. ${\gamma_\text{D}}$ can be simplified according to 
\begin{eqnarray}
\label{eq:hR-sumH}
 \gamma_\text{D} =\frac{ \gamma_\text{S} \sum\limits_{i=1}^{n_\text{S}} \mid \boldsymbol{h}_{2w}\boldsymbol{\Lambda}_{G}^{\frac{1}{2}}\boldsymbol{\Lambda}_{\Sigma}\boldsymbol{h}_{1w,i}\mid^2}{n_\text{S}(1+ \boldsymbol{h}_{2w}\boldsymbol{\Lambda}_G\boldsymbol{\Lambda}_\Sigma\boldsymbol{h}^H_{2w})}
\end{eqnarray}
where ${\boldsymbol{h}_{1w,i}}$ represents the ${i}$th column of
${\boldsymbol{H}_{1w}}$. Let us assume
\begin{subequations}
  \begin{align}
\label{eq:Xj}
X_j&=|h_{2w,j}|^2\\
\label{eq:Y}
Y&=\frac{1}{n_\text{S}}\sum_{i=1}^{n_\text{S}} |h_{1w,i}|^2
  \end{align}
\end{subequations}

It is proved in \cite[Appendix~1]{MoGo:2014TWC} that ${ \gamma_\text{D}}$ in (\ref{eq:hR-sumH}) can be further simplified to 
\begin{eqnarray}
\label{eq:hR-XY}
 \gamma_\text{D}= \gamma_\text{S} Y \times \underbrace{\frac{\sum\limits_{j=1}^{\kappa} {\lambda}_{j}^{G}{\lambda}_{j}^{\Sigma 2}  X_j }{1+ \sum\limits_{j=1}^{\kappa} {\lambda}_{j}^{G}{\lambda}_{j}^{\Sigma}  X_j }}_{X}
\end{eqnarray} 
where $\kappa$ is the minimum of $n_\text{R}$ and Number of Non-Zero (nnz) $\lambda_{j}^{\Sigma}$, i.e.,
\begin{eqnarray}
\label{eq:kappa}
\kappa = \mathrm{min}(n_\text{R}, \text{nnz}(\lambda_{j}^{\Sigma}))
\end{eqnarray}
Note that $Y$ and  $X_j$ correspond to the S-R and R-D  channels, respectively. The random variable
\begin{equation}
\label{eq:X}
X= \frac{\sum\limits_{j=1}^{\kappa} {\lambda}_{j}^{G}{\lambda}_{j}^{\Sigma 2}  X_j }{1+ \sum\limits_{j=1}^{\kappa} {\lambda}_{j}^{G}{\lambda}_{j}^{\Sigma}  X_j }
\end{equation}
in (\ref{eq:hR-XY}) incorporates the effect of the R-D link as well as the effect of power allocation due to $\lambda^G_j$.
Furthermore, since we assume Rayleigh fading in both the S-R and R-D links, hence, ${X_j}$ is exponentially distributed with unit mean and ${Y}$ has an Erlang-distribution with rate and shape parameters equal to ${n_\text{S}}$.

Although it is proposed in \cite{MoGo:2014TWC} that the optimal $\boldsymbol{G}$ should be diagonalized according to (\ref{eq:G-Decomposition}) where $\boldsymbol{\Lambda}_G$ is a diagonal matrix with its components organized in descending order, an optimal power allocation method to distribute relay's transmit power among different ${\lambda^G_j}$s is not discussed. That is still an open problem but will be addressed in this paper. 

\subsection{Contribution}
\label{Sec:Contribution}

There are two main problems investigated in the following sections, each leading to novel contributions:
\begin{itemize}

\item We evaluate, for the first time,  the statistical characteristics of the ${\gamma_\text{D}}$ introduced in 
(\ref{eq:hR-XY}). Due to its mathematical complexity,  the exact pdf of ${\gamma_\text{D}}$ is not derived but an approximation to the pdf is provided in 
this work. The  approximated pdf is then used for 
calculating the outage probability and it is 
illustrated, by the simulations, that the  approximated cdf leads to
rather accurate results. Moreover, the exact cdf of ${\gamma_\text{D}}$ will be derived for the two
asymptotic scenarios of full-correlation and no-correlation at the relay.  Although this novel cdf  is helpful for outage analysis of the system, it is too complicated to be used for the analysis of the ergodic capacity.

\item In order to approximate the maximum achievable transmission rate, a very simple power allocation algorithm at the relay is introduced in this paper. As the optimal power allocation at moderate and high SNR is still an open problem\footnote{The optimal power allocation at low SNR was proposed in \cite{MoGo:2014TWC}.}, for the purpose of comparison, exhaustive search over various discrete values of the rates is used as a benchmark. The values of the rates are obtained by allocating various amount of the power among different eigenmodes. According to the simulations, the proposed power allocation algorithm approximates the benchmark with insignificant difference. 
\end{itemize}

\begin{figure}[t]
\begin{center}
        \includegraphics[width=0.45\textwidth, height =0.43 \textwidth]{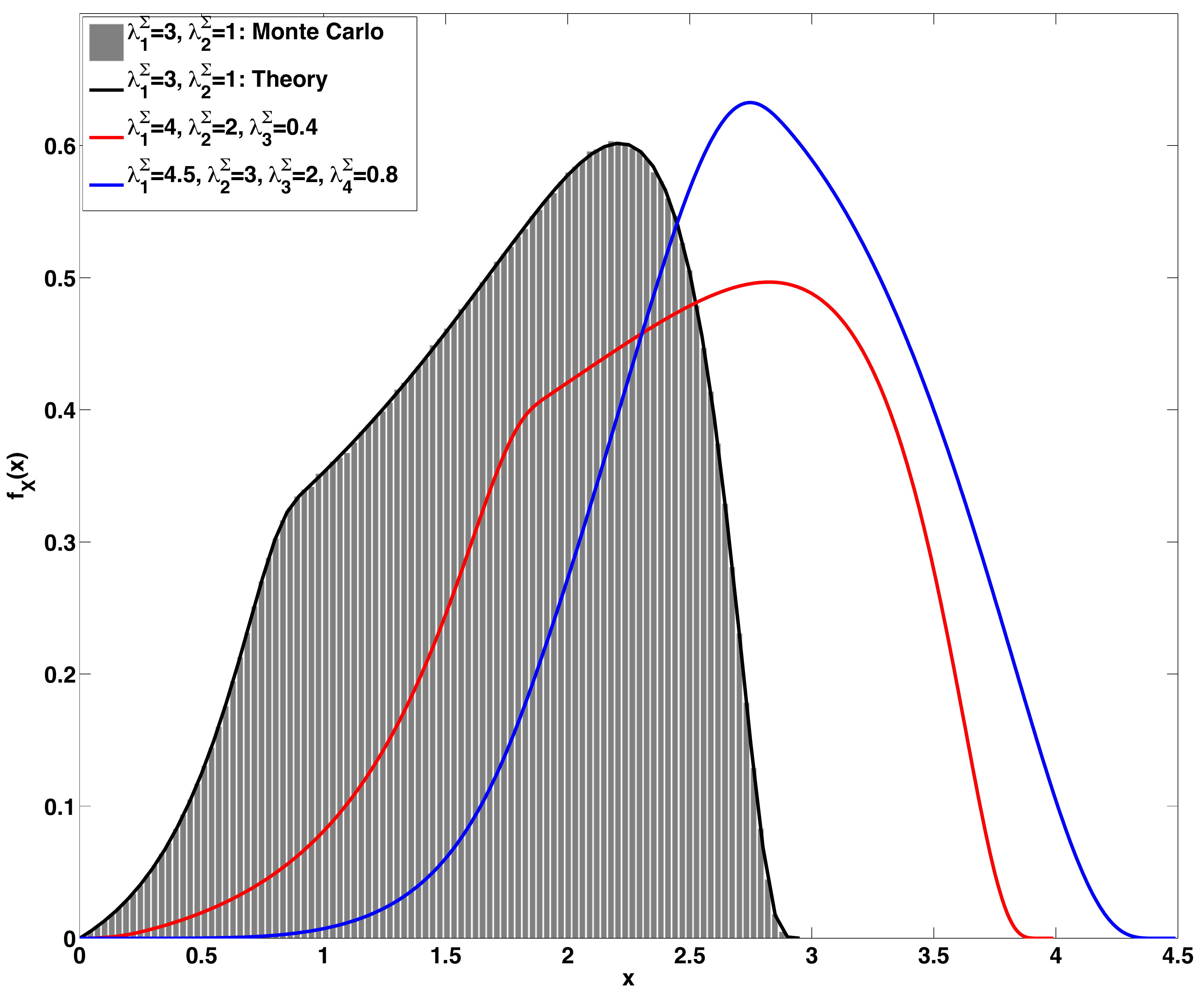}
\end{center} 
\caption{Pdf of $X$ with various number of $\lambda^{\Sigma}_j$ parameters. Monte Carlo simulations validate the correctness of the theoretical results.}
\label{fig:PDF_X}
\end{figure} 
\section{Statistical Analysis of Received SNR at Destination} 
\label{Sec:Distribution of Received SNR at Destination}
The SNR distribution in the relay is directly related to the ergodic capacity of the system evaluated in this paper. In order to maximize ${\mathbb{E}\left\lbrace C(\cdot) \right\rbrace}$ in (\ref{eq:ErgodicCapactityDefinition}), one should distribute the available relay power appropriately among  different ${\lambda^G_j}$s in (\ref{eq:hR-XY})  so that ${\mathbb{E}\left\lbrace C(\cdot) \right\rbrace}$ is maximized. Before we continue with the statistical characterization of $\gamma_\text{D}$, two extreme scenarios are studied: we will investigate $\gamma_\text{D}$ when the relay does not experience any fading correlation and also when the antennas in the relay are fully correlated. These two scenarios, indeed, provide performance bounds, and so, the performance of a system with partial correlation will fall between the two bounds.

\subsection{Statistical Characteristics of $\gamma_\text{D}$ Assuming Full Correlation}
\label{Sec:full Correlation}
Full correlation~(FC) at the relay is equivalent to considering a system where all the elements of the $\boldsymbol{\Sigma}$ are unity\footnote{We assume normalized correlation for simplicity of the notation. Generalization to a case with non-unity full correlation is straightforward. Similar normalization is assumed in Section~\ref{Sec:No Correlation} too.}, i.e., ${\boldsymbol{\Sigma} =\boldsymbol{1}_{n_\text{R}\times n_\text{R}}}$; consequently, it is easy to see  that ${\lambda^{\Sigma}_1 = n_\text{R}}$ and $\lambda^{\Sigma}_j = 0$ for $j\geq 2$, and so, the random variable $X$ in \eqref{eq:X} can be written as 
\begin{eqnarray}
\label{eq:X-FC}
X_{FC} = \frac{\lambda^{G}_1 \lambda^{\Sigma 2}_1 X_1 }{1+ \lambda^{G}_1 \lambda^{\Sigma }_1 X_1} =  
\frac{n^2_\text{R} \lambda^{G}_1 X_1 }{1+ n_\text{R} \lambda^{G}_1 X_1} = n_\text{R} \frac{V}{1+ V}
\end{eqnarray}
where, the random variable ${V=n_\text{R} \lambda^{G}_1 X_1}$ is exponentially distributed with mean ${n_\text{R} \lambda^{G}_1}$.
The following theorem introduces the cdf of $\gamma_\text{D}$ for a fully correlated relay:
\begin{thm}
Assuming fully correlated antennas at the relay, the cdf of $\gamma_\text{D}$ is
\begin{eqnarray}
\label{eq:CDF-gammaD-FC}
F_{\gamma_{\text{D}-{FC}}}(x)&=& 1-2(n_\text{S}w)^{n_\text{S}}\mathrm{e}^{-n_\text{S}w}\\
&& \hspace*{-10mm}\times\sum\limits_{m=0}^{n_\text{S}-1}
\frac{(\lambda_1^G n_\text{S}n_\text{R}w)^{-(m+1)/2}}{m! (n_\text{S}-m-1)!}\mathrm{K}_{m+1}(2\sqrt{\frac{n_\text{S}w}{\lambda_1^G n_\text{R}}})
\nonumber
\end{eqnarray}
\end{thm}
where $\mathrm{K}_{\nu}(\cdot)$ is the modified Bessel function of the second kind, ${\nu\text{-th}}$ order and $w = \frac{x}{n_\text{R}\gamma_\text{S}}$. 
\begin{proof}
See Appendix~\ref{app:SNR Distribution: Full Correlation Scenario} for a detailed proof.
\end{proof}
The subscript ``FC" in $F_{\gamma_{\text{D}-{FC}}}(x)$ indicates the Full Correlation scenario. One can easily derive the pdf of $\gamma_\text{D}$ for the full correlation scenario by taking the derivative of \eqref{eq:CDF-gammaD-FC} with respect to $x$.

In the next section, we study the statistical characteristics of the system for the non-correlated fading scenario.

\subsection{Statistical Characteristics of $\gamma_\text{D}$ Assuming No Correlation}
\label{Sec:No Correlation}
No Correlation~(NC) at the relay translates to  $\boldsymbol{\Sigma}=\boldsymbol{I}$. Such a scenario will occur when the relay antennas are placed sufficiently far apart and that the relay node is placed in a rich scattering environment. Assuming $\boldsymbol{\Sigma}=\boldsymbol{I}$, it is straightforward to conclude that $\lambda_1^{\Sigma}=\lambda_1^{\Sigma}=\cdots =\lambda_{n_\text{R}}^{\Sigma} =1$. On the other hand, since all $X_j$ random variables follow the same distribution~(exponential distribution with unit mean) and as all $\lambda_j^{\Sigma}$ values are equal to one, therefore all the $\lambda_j^{G}$ values should be assigned the same power, and so, let us assume ${\lambda_1^{G}=\lambda_2^{G}=\cdots =\lambda_{n_\text{R}}^{G} =\lambda_{\text{eq}}^{G}}$; consequently, the random variable $X$ will be written as
\begin{eqnarray}
\label{eq:X-NC}
X_{NC} = \frac{\lambda_{\text{eq}}^{G} \sum_{j=1}^{n_\text{R}}X_j }{1+ \lambda_{\text{eq}}^{G} \sum_{j=1}^{n_\text{R}}X_j}=\frac{V}{1+ V}
\end{eqnarray}
where $V$ follows the Erlang distribution with rate $\frac{1}{\lambda_{\text{eq}}^{G}}$ and shape~$n_\text{R}$.
By substituting $X_{NC}$ in \eqref{eq:hR-XY}, the following expression will be derived for the cdf of $\gamma_\text{D}$ for the non correlated scenario:
\begin{thm}
Assuming no correlation at the relay, the cdf  of $\gamma_\text{D}$ is
\begin{eqnarray}
\label{eq:CDF-NC}
F_{\gamma_{\text{D}-{NC}}}(x)&=& 1-2(n_\text{S}w)^{n_\text{S}}\mathrm{e}^{-n_\text{S}w}\\
&& \hspace*{-20mm}\times\sum_{m=0}^{n_\text{R}-1}\sum_{n=0}^{n_\text{S}-1} 
\frac{n_\text{S}^m
(\lambda_{\text{eq}}^{G} n_\text{S}w)^{-\frac{m+n+1}{2}}
}{m!n!(n_\text{S}-n-1)
!w^m}
\mathrm{K}_{m-n-1}(2\sqrt{\frac{n_\text{S}w}{\lambda_{\text{eq}}^{G}}})
\nonumber
\end{eqnarray}
\end{thm}
with $w = \frac{x}{\gamma_\text{S}}$.
\begin{proof}
Following the same lines of the proof for \eqref{eq:CDF-gammaD-FC}, one can easily prove \eqref{eq:CDF-NC} too. 
\end{proof}

Assuming single antennas at the source and the relay nodes, i.e., ${n_\text{S} = n_\text{R} =1}$, the MIMO scenario of this paper reduces to a conventional single antenna relay system wherein the relay node has access to the variance of its channels; consequently, as expected, the two expressions in \eqref{eq:CDF-gammaD-FC} and \eqref{eq:CDF-NC} are identical according to 
\begin{eqnarray}
\label{eq:single antenna}
F_{\gamma_{\text{D}}}(x)=
1-2\sqrt{\frac{(1+\gamma_\text{S})x}{\gamma_\text{S}\gamma_\text{R}}}
\mathrm{e}^{- \frac{x}{\gamma_\text{S}}}
\mathrm{K}_{1}(2\sqrt{\frac{(1+\gamma_\text{S})x}{\gamma_\text{S}\gamma_\text{R}}})
\end{eqnarray} 
Note that $F_{\gamma_{\text{D}}}(x)$ for a single antenna scenario similar to \eqref{eq:single antenna} has been reported in numerous papers including, e.g.,  \cite{HaAl:2003ICASSP}.

As discussed earlier, $F_{\gamma_\text{D}}$~(and $f_{\gamma_\text{D}}$) for arbitrary correlation  should follow an expression that at the extreme case reduces to \eqref{eq:CDF-NC} and \eqref{eq:CDF-gammaD-FC}.

On the other hand, statistical characteristics of $X$ are also of essential importance for understanding the distribution $\gamma_\text{D}$.  Assuming full correlation and no correlation at the relay, characterizing $X$ was simple, however, characterizing $X$ with arbitrary correlation is, actually, more complicated and will be derived in the next section.

\subsection{Statistical Characteristics of $X$ for Arbitrary Correlation}
\label{Sec:PDF-X}

The statistical characteristics of the random variable $X$ are not studied in the literature but will be derived in this paper. 

\begin{figure}[t]
\begin{center}
        \includegraphics[width=0.43\textwidth, height =0.40 \textwidth]{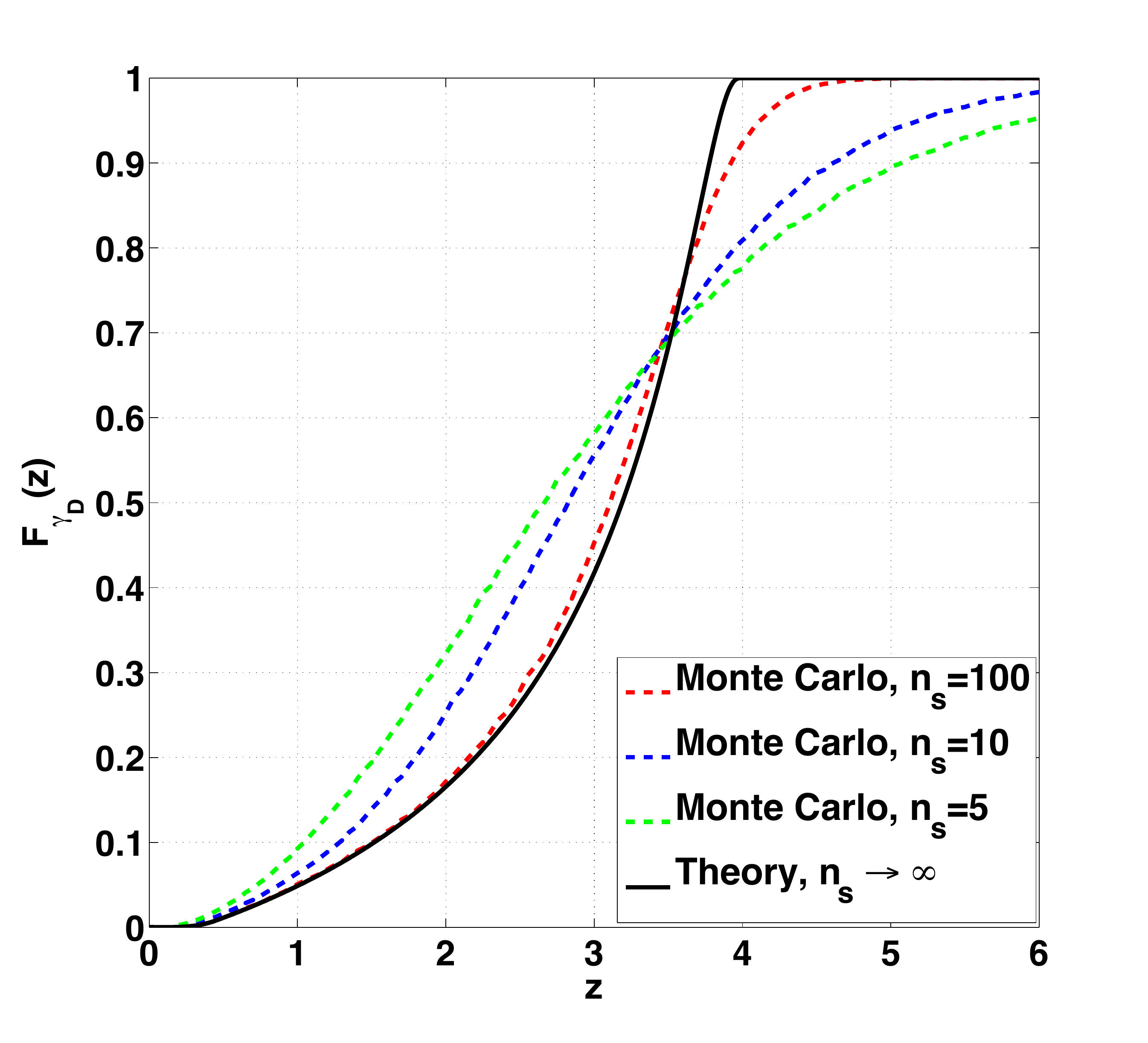}
\end{center} 
\caption{{ Approximation of $F_{\gamma_D}(z)$ as obtained in \eqref{thm:F-gamma=F-X} where $\lambda_1^{\Sigma}=4$, 
$\lambda_2^{\Sigma}=1$ and 
$\lambda_1^{G}=\lambda_2^{G}=\lambda_\text{eq}^{G}=2$.} }  
\label{fig:CDF}
\end{figure} 
Considering that ${\lambda^{\Sigma}_{1}\ge \lambda^{\Sigma}_{2}\ge \cdots \geq \lambda^{\Sigma}_{\kappa}\geq 0}$, the cdf of $X$ in (\ref{eq:X}) is 
\begin{eqnarray}
\label{eq:cdf-X}
F_X(x) = \begin{cases}
0& \text{$x\leq 0$}\\
1-\sum\limits_{i=1}^{j}\prod\limits_{{m=1}\atop {m\neq j}}^{\kappa} \frac{c_m}{c_m - c_j}\mathrm{e}^{c_i x} & \lambda^{\Sigma}_{j+1} < x < \lambda^{\Sigma}_{j}\\
1& \text{$\lambda^{\Sigma}_{1} \leq x$}
\end{cases} 
\end{eqnarray}
where 
\begin{equation}
\label{eq:c-j}
c_j=\frac{1}{\lambda^{G}_j\lambda^{\Sigma}_j(\lambda^{\Sigma}_j - x)}.
\end{equation}
A sketch of the proof for  \eqref{eq:cdf-X} is provided in Appendix~\ref{app:cdfX}.
To provide a better understanding of the statistical characteristics of the random variable $X$, the $F_X(x)$ for ${\kappa =2\text{ and }3}$ is provided in the following. Assuming $\kappa=2$, we have 
 \begin{eqnarray}
\label{eq:cdf-kappa=2}
F_X(x) = \begin{cases}
0 &x\leq 0\\
1-\frac{c_2}{c_2-c_1}\mathrm{e}^{c_1 x} - \frac{c_1}{c_1-c_2}\mathrm{e}^{c_2 x}
 & 0 < x < \lambda^{\Sigma}_{2}\\
1-\frac{c_2}{c_2-c_1}\mathrm{e}^{c_1 x}
 & \lambda^{\Sigma}_{2} < x < \lambda^{\Sigma}_{1}\\
1& x \geq \lambda^{\Sigma}_{1}
\end{cases} 
\end{eqnarray}
and for $\kappa=3$, $F_X(x)$ is derived in (\ref{eq:cdf-kappa=3}) at the top of next page.
\newcounter{mytempeqncnt}
\begin{figure*}[!t]

\normalsize
\setcounter{mytempeqncnt}{\value{equation}}

 \begin{eqnarray}
\label{eq:cdf-kappa=3}
F_X(x) = \begin{cases}
0&x\leq 0\\
1-\frac{c_2 c_3}{(c_2-c_1)(c_3-c_1)}\mathrm{e}^{c_1 x} - \frac{c_1 c_3}{(c_3-c_2)(c_1-c_2)}\mathrm{e}^{c_2 x}
-\frac{c_1 c_2}{(c_1-c_3)(c_2-c_3)}\mathrm{e}^{c_3 x} &0 < x < \lambda^{\Sigma}_{3}
\\
1-\frac{c_2 c_3}{(c_2-c_1)(c_3-c_1)}\mathrm{e}^{c_1 x} - \frac{c_1 c_3}{(c_3-c_2)(c_1-c_2)}\mathrm{e}^{c_2 x} &\lambda^{\Sigma}_{3}< x < \lambda ^{\Sigma}_{2}\\
1-\frac{c_2 c_3}{(c_2-c_1)(c_3-c_1)}\mathrm{e}^{c_1 x} &\lambda^{\Sigma}_{2}< x < \lambda^{\Sigma}_{1}\\
1& x \geq \lambda^{\Sigma}_{1}
\end{cases} 
\end{eqnarray}
\hrulefill
\vspace*{4pt}
\end{figure*}
One can easily derive ${f_X(x)}$ by taking the derivative of ${F_{X}(x)}$ with respect to $x$, i.e.,
\begin{eqnarray}
\label{eq:PDF-X}
f_{X}(x) = \frac{\mathrm{d}F_{X}(x)}{\mathrm{d}x}
\end{eqnarray}
 which, is a straightforward simple derivation practice.  Fig.~\ref{fig:PDF_X} illustrates the ${f_{X}(x)}$ assuming 
 $\lambda^G_j =1$ for various values of $\lambda^\Sigma_j$. The agreement between the theoretical and Monte Carlo simulations validates the correctness of the calculations. 
\begin{figure}[t]
\begin{center}
        \includegraphics[width=0.45\textwidth, height =0.41 \textwidth]{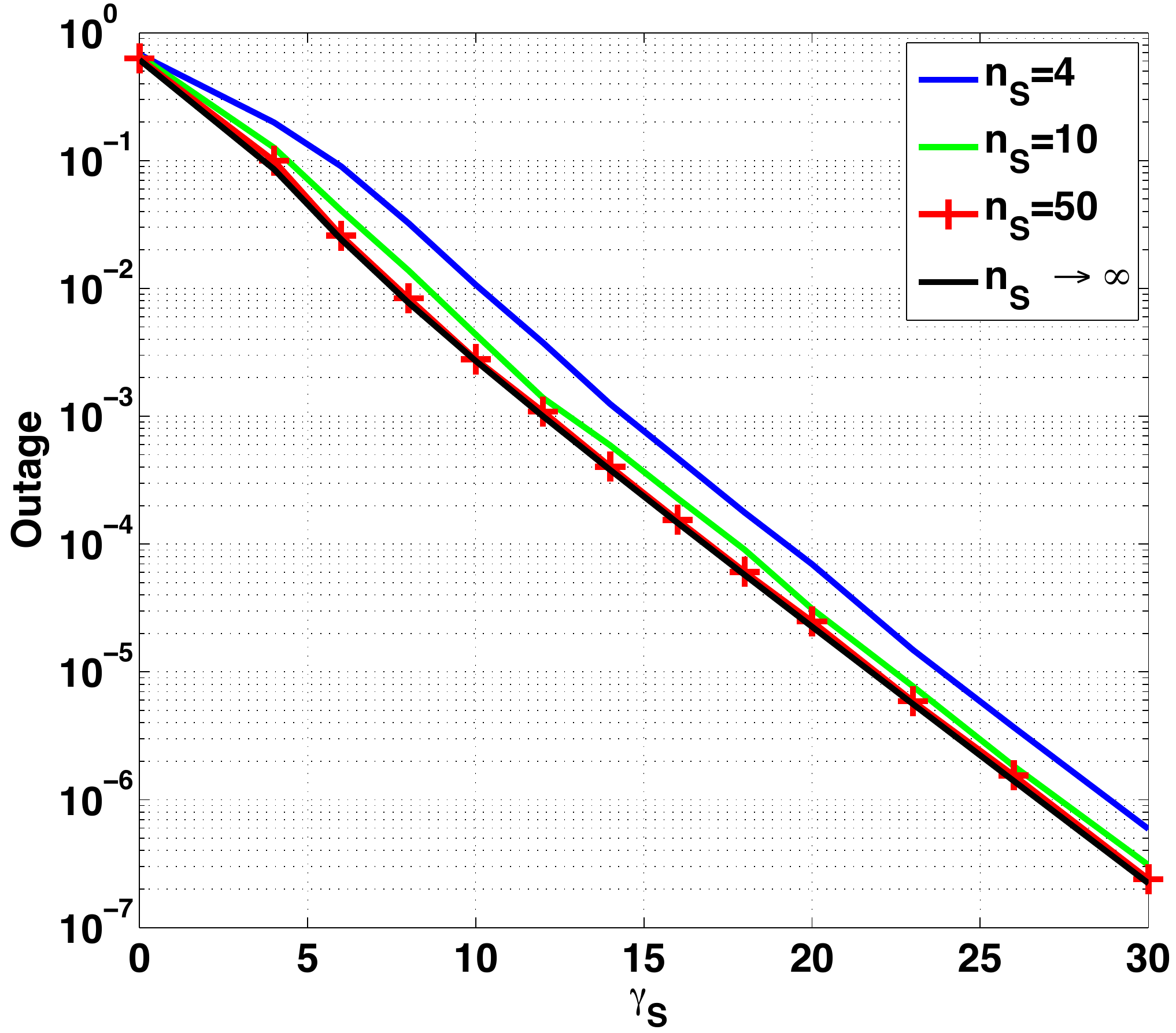}
\end{center} 
\caption{Outage probability approximated by  \eqref{thm:F-gamma=F-X}~($n_\text{S}\to \infty$) in comparison with the Monte Carlo simulations for various values of $n_\text{S}$.
 }
\label{fig:outage}
\end{figure} 
\subsection{Statistical Characteristics of $\gamma_\text{D}$ Assuming Infinite Antennas at the Source}
From the discussions provided in the previous section,
it is clear that ${\gamma_\text{D} = \gamma_\text{S}YX}$ where $Y$ follows an Erlang distribution with rate and shape parameters equal to $n_\text{S}$, i.e., ${f_{Y}(y)= \frac{n_\text{S}^{n_\text{S}}}{({n_\text{S}}-1)!}y^{n_\text{S}-1}\mathrm{e}^{-n_\text{S}y}}$  and $X$ follows a distribution as derived in (\ref{eq:cdf-X}) and (\ref{eq:PDF-X}). 
In order to calculate the cdf of $\gamma_\text{D}$, one should calculate the following integral
\begin{eqnarray}
\label{eq:cdf_gammaD-integral}
F_{\gamma_\text{D}}(z) &=& \mathbb{P}\{\gamma_\text{D}<z \}=\mathbb{P}\{\gamma_\text{S}YX<z \}\nonumber \\
&=&\int_{0}^{\infty}\underbrace{\mathbb{P}\{X<\frac{z}{y\gamma_\text{S}}\}}_{F_X(\frac{z}{y\gamma_\text{S}})}f_Y(y)\mathrm{d}y.
\end{eqnarray}
On the other hand, (\ref{eq:cdf_gammaD-integral}) does not lend itself easily to further calculations, and to the best of our knowledge, the integral cannot be solved with the existing table of integrals~(e.g., \cite{GrRy:2007}). However, clearly, $F_{\gamma_\text{D}}(z)$ is a multipartite function because $X$ has a multipartite cdf; to the best of our knowledge, a multipartite $F_{\gamma_\text{D}}(z)$ in the context of AF cooperative systems has not been reported in the literature and so it is observed for the first time in this paper\footnote{ We use multipartite function to refer to a function that involves several distinct functions for different domains~(e.g., see \eqref{eq:cdf-kappa=2} and \eqref{eq:cdf-kappa=3}). Note that being multipartite is not considered to be advantage~(or disadvantage) for the system but stressing on the novelty of the statistical characteristics of the SNR, in the context of AF cooperative systems, is meant to highlight the need for further investigation on the problem.}. 
The following theorem provides an approximate $F_{\gamma_\text{D}}(z)$ for the given system model:
\begin{thm}
For large $n_\text{S}$, $F_{\gamma_\text{D}}(z)$ can be approximated by
\begin{equation}
\label{thm:F-gamma=F-X}
F_{\gamma_\text{D}}(z) \approx F_{X}(\frac{z}{\gamma_\text{S}})
\end{equation}
\end{thm}
 \begin{proof}
 The random variable ${Y}$ in the previous sections, e.g. in (\ref{eq:hR-XY}), follows an Erlang-distribution; indeed, ${Y}$ is the sum of ${n_\text{S}}$ exponential random variables, each with parameter ${n_\text{S}}$~(see (\ref{eq:Xj})). Using the central limit theorem (see \cite[Ch. 7.4]{PaPi:2002}), the random variable ${Y}$ can be approximated by a Gaussian distribution
with mean equal to ${1}$ and variance ${1/n_\text{S}}$, i.e., approximately,
${Y \sim \mathcal{N}(1,1/n_\text{S}}$). Assuming large ${n_\text{S}}$ (i.e., ${n_\text{S} \to \infty}$), one can easily deduce ${Y\to 1}$. By setting $Y=1$ in (\ref{eq:cdf_gammaD-integral}), one can write 
${F_{\gamma_\text{D}}(z) = \mathbb{P}\{\gamma_\text{D}<z \}=\mathbb{P}\{X<\frac{z}{\gamma_\text{S}} \}}$ and so (\ref{thm:F-gamma=F-X}) is proved.
 \end{proof}
{Fig.~\ref{fig:CDF} and Fig.~\ref{fig:outage} are intended to validate the precision of the approximation obtained in \eqref{thm:F-gamma=F-X}. In  
Fig.~\ref{fig:CDF} an illustration of $F_{\gamma_D}(z)$ is provided; clearly, Monte Carlo simulations approximate theoretical $F_{\gamma_D}(z)$ when $n_\text{S}$ is large. Moreover, considering that
the cdf $F_{\gamma_\text{D}(z)}$ in (\ref{thm:F-gamma=F-X}) corresponds to the outage probability for large $n_\text{S}$~(ideally for ${n_\text{S}\to \infty}$), in Fig.~\ref{fig:outage} we plot outage probability versus transmit power at the source node using (\ref{thm:F-gamma=F-X}) and also using Monte Carlo simulations for various values of $n_\text{S}$. It is clear that for large values of $n_\text{S}$, the closed form expression for the outage probability approximates the Monte Carlo simulations with high accuracy. Nevertheless, for smaller values of the $n_\text{S}$, although the approximation is not accurate, it provides a reasonable approximation.}

As mentioned in \ref{Sec:Contribution}, one of the main 
objectives in this paper is to allocate available power in 
the relay according to (\ref{eq:RelayPowConstraintLambda}) 
among different ${\lambda_j^{G}}$ variables so that the 
ergodic capacity in (\ref{eq:ErgodicCapactityDefinition}) is 
maximized. In fact, for plotting Fig.~\ref{fig:PDF_X} we 
arbitrarily assumed ${\lambda_1^{G}=\lambda_2^{G}=1}$; 
however, such a random power allocation to ${\lambda_1^{G}}$ 
and ${\lambda_2^{G}}$ does not guarantee that the 
maximization problem in 
(\ref{eq:ErgodicCapactityDefinition}) is solved. On 
the other hand, it was observed in \eqref{eq:cdf_gammaD-integral} that the $F_X(\cdot)$ expression derived in \eqref{eq:cdf-X} is 
too complicated to lend itself to further mathematical 
calculations. Indeed, we do not know any closed form 
expression for the objective function 
$\mathbb{E}\{C(\cdot)\}$ which can be used for calculating optimal ${\lambda_j^{G}}$ 
values in (\ref{eq:ErgodicCapactityDefinition}). Furthermore, not only do we not know 
any analytical way for calculating optimal ${\lambda_j^{G}}$ values, we are not aware 
of any numerical method to calculate optimal ${\lambda_j^{G}}$ values. {For the 
purpose of comparison, exhaustive search over various discrete values  of the achievable rates is used as a benchmark. The 
rates for the exhaustive search are obtained by assigning various amount of the power 
to the eigenmodes according to the power constraint in 
\eqref{eq:RelayPowConstraintLambda}; {moreover, the resolution of the exhaustive search is kept adequately small~($0.1$~dB) to ensure accurate approximation. Note that resolutions larger than $0.1$~dB also provide accurate results, however, to make sure that no local maximum is missed, we use the the resolution of $0.1$~dB throughout the paper when exhaustive search is provided for comparison.}

  Although the proposed method in the next section is simple and straightforward, it will be revealed that the obtained values for ${\lambda_j^{G}}$s lead to the reasonable rates that are indistinguishable from the benchmark rates. Also, it will be revealed that the proposed method significantly reduces the computationally expensive calculations due to exhaustive search for finding optimal  ${\lambda_j^{G}}$ values in real time practical communication systems.
  
\section{Two Antenna Relay}
\label{Sec:Two Antennas in the Relay}

For simplicity, as an initial step, let us assume a system with two antennas at the relay~(i.e.~${n_\text{R}=2}$) where
${
\boldsymbol{\Sigma} = 
[{\substack{1\;\; \;\rho \atop
\rho^* \;1}}]
}$; the parameter $\rho$ indicates the correlation coefficient. As there are only two eigenvectors corresponding to $\boldsymbol{\Sigma}$, the problem of  optimal power allocation reduces to calculating the optimal values of ${\lambda_1^G}$ and ${\lambda_2^G}$, given the power constraint in (\ref{eq:RelayPowConstraintLambda}).

In \cite[Eq.~34]{MoGo:2014TWC}, we derive a necessary and sufficient condition under which transmission \emph{only} from the largest eigenvector (the eigenvector corresponding to ${\lambda_1^{\Sigma}}$) achieves capacity. For ease of reference, \cite[Eq.~34]{MoGo:2014TWC} is provided in the following lemma:

\emph{\textbf{Lemma}: Transmission from the largest eigenvector achieves capacity if}
\begin{eqnarray}
\label{eq:LER-Constraint}
{\lambda}_{2}^{\Sigma} \leq
\frac{(\alpha_1+P_1\lambda^{\Sigma}_2)\mathcal{D}(\lambda^{\Sigma}_1,P_1) -\alpha_1\mathbb{E}\lbrace\frac{1}{1+P_1Z_1} \rbrace   }
{ P_1\mathbb{E}\lbrace \frac{1+\lambda^{\Sigma}_2\gamma  Y}{1+P_1
Z_1}\rbrace}
\end{eqnarray} 
with
\begin{eqnarray}
\label{eq:ZandD}
Z_1&=&{\lambda}_{1}^{\Sigma} ( 1+\gamma  {\lambda}_{1}^{\Sigma} Y)X_1\\
\mathcal{D}({\lambda}_{1}^{\Sigma},P_1)&=&\frac{1}{P_1{\lambda}_{1}^{\Sigma}}\Gamma(0,\frac{1}{P_1{\lambda}_{1}^{\Sigma} })\mathrm{e}^{\frac{1}{P_1{\lambda}_{1}^{\Sigma} }}
\end{eqnarray}  
where $P_1=\frac{P_\text{R}}{\lambda^{\Sigma}_1P_\text{S}+N_0}$ and  $\alpha_1=\frac{\lambda^{\Sigma}_2P_\text{S}+N_0}{\lambda^{\Sigma}_1P_\text{S}+N_0}$.
 Although \cite{MoGo:2014TWC} proves that LER is the optimal transmission method at low SNR, it does not discuss any method to distribute the available power at the relay among ${\lambda_j^G}$ variables when (\ref{eq:LER-Constraint}) does not hold. \emph{This problem is addressed in the rest of the paper.}
\begin{figure}[t]
\begin{center}
        \includegraphics[width=0.43\textwidth, height =0.41 \textwidth]{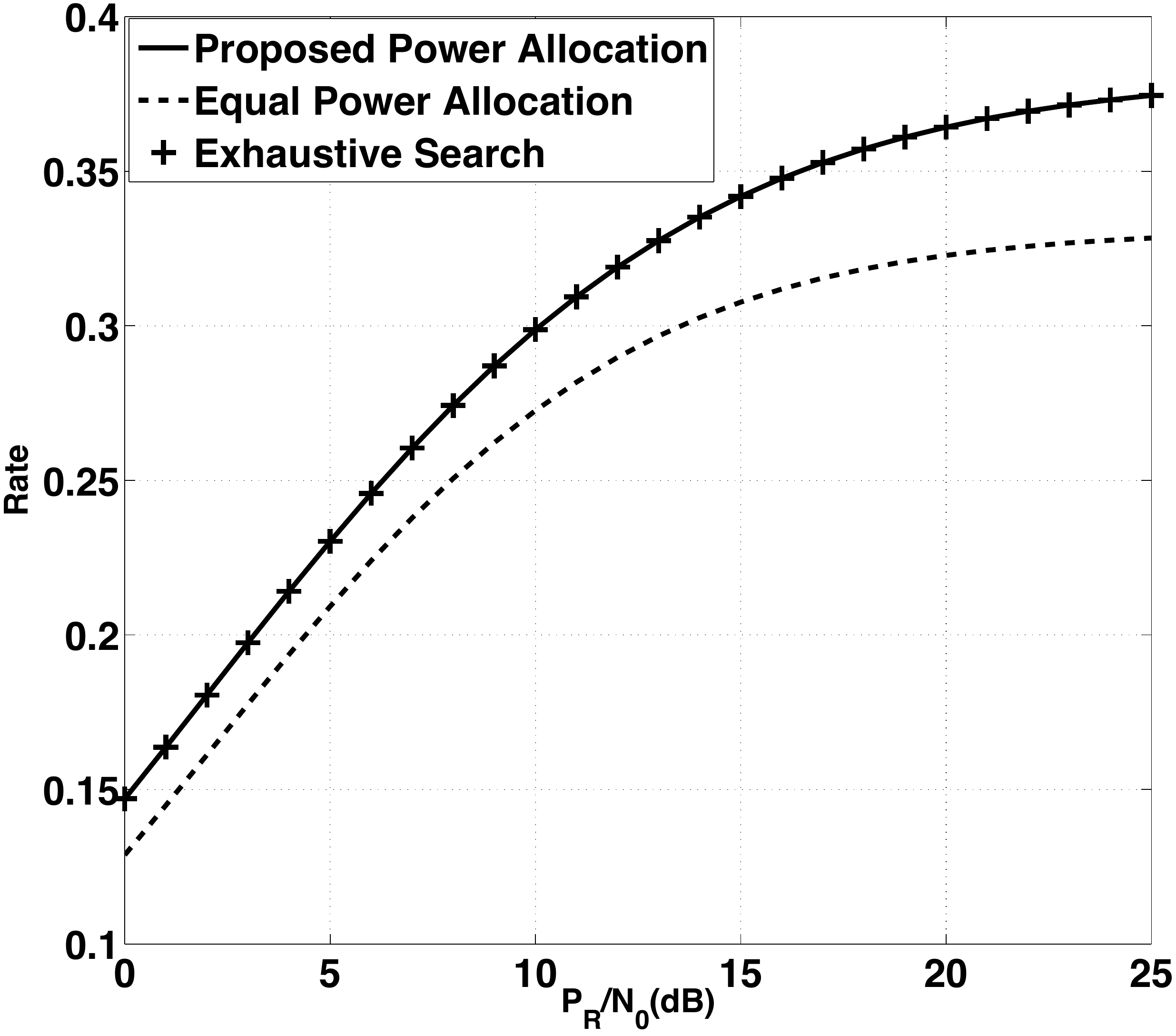}
\end{center} 
\caption{ {Rate vs. transmit power at the relay. The rate values are obtained using the proposed power allocation algorithm, equal power allocation and the benchmark. $n_\text{S} = n_\text{R}=2$, ${N_0 = 1}$, ${P_\text{S}=0~\text{dB}}$, and inter-antenna correlation ${\rho = 0.3}$. The exhaustive search is performed over discrete values of transmission rates that were obtained using various values of ~${\lambda^G_1\text{ and }\lambda^G_2}$~(step size $0.1$~dB) that fulfil~(\ref{eq:RelayPowConstraintLambda}).}
 }
\label{fig:C_2RelAnt}
\end{figure} 
\paragraph*{\textbf{Proposition}}
 Given  $\lambda_1^{\Sigma}$, $\lambda_2^{\Sigma}$, $P_\text{S}$ and $P_\text{R}$ , if  (\ref{eq:LER-Constraint}) holds, allocate the entire power in the relay only to ${\lambda_1^G}$ as
\begin{eqnarray}
\label{eq:LowSNR-Lambda1G-2antenna}
\lambda_1^G = \frac{P_\text{R}}{N_0+P_\text{S}\lambda_1^{\Sigma}}
\end{eqnarray} 
and set ${\lambda_2^{G}=0}$~(i.e., LER), otherwise, when (\ref{eq:LER-Constraint}) does \emph{not} hold, we propose to allocate power per eigenvector proportionally to the strength of the eigenmodes, i.e.,
\begin{eqnarray}
\label{eq:proportional power}
\frac{\lambda_1^G}{\lambda_2^G} =\frac{\lambda_1^{\Sigma}}{\lambda_2^{\Sigma}}.
\end{eqnarray}
 Consequently, one can assume ${\lambda_1^G =\lambda_1^{\Sigma} g_2}$ and ${\lambda_2^G =\lambda_2^{\Sigma} g_2}$
with $g_2$ obtained from \eqref{eq:RelayPowConstraintLambda} as
\begin{eqnarray}
\label{eq:g2}
g_2=\frac{ P_\text{R}}{P_\text{S}(\lambda_1^{\Sigma 2} + \lambda_2^{\Sigma 2}) + N_0(\lambda_1^{\Sigma} + \lambda_2^{\Sigma })}.
\end{eqnarray}  
Note that \eqref{eq:LowSNR-Lambda1G-2antenna} and \eqref{eq:g2} are obtained from the relay power constraint in  (\ref{eq:RelayPowConstraintLambda}).\hfill $\blacksquare$

{\textbf{\textit{Remark:}} The motivation for assuming proportional power allocation in \eqref{eq:proportional power} arises from the limit behaviour of the correlation coefficients. In the case where $\lambda_1^{\Sigma}$ much larger than $\lambda_2^{\Sigma}$~(i.e., $\lambda_1^{\Sigma}\gg \lambda_2^{\Sigma}$) clearly LER will be optimal transmission method and so $\lambda_1^{G}\gg \lambda_2^{G}$ must hold, and this is guaranteed by \eqref{eq:proportional power}. On the other hand when $\lambda_1^{\Sigma}$ and $\lambda_2^{\Sigma}$ are only slightly different, both $\lambda_1^{G}$ and $\lambda_2^{G}$ should be assigned relatively equal power; indeed, when $\lambda_1^{\Sigma}=\lambda_2^{\Sigma}$, the fading is uncorrelated and so, as described in Subsection \ref{Sec:No Correlation}, $\lambda_1^{G}=\lambda_2^{G}=\lambda_\text{eq}^{G}$,
 which again is guaranteed by \eqref{eq:proportional power}.Note that the conjecture will be validated in the following by simulations which show that the result is nearly identical to that with power allocation by exhaustive search. }

 {Fig.~\ref{fig:C_2RelAnt} illustrates the transmission rates of a cooperative system with $n_\text{S}=n_\text{R}=2$ when the power allocation is carried out using the proposed algorithm. For comparison, maximum transmission rates corresponding to the benchmark are also illustrated.} The figure clearly shows a good agreement between exhaustive search~(i.e., the benchmark) and also the simple proposed algorithm.   
The difference between the proposed algorithm and the benchmark is, in fact, indistinguishable. Fig.~\ref{fig:C_2RelAnt} shows the transmission rate assuming equal power allocation in the relay. Note  that equal power transmission is equivalent to ignoring the knowledge of correlation at the relay. Clearly, the proposed algorithm significantly outperforms equal power transmission.

With two antennas at the relay, Fig.~\ref{fig:C_2RelAnt} shows that the proposed algorithm leads to excellent results. In the next section, the proposed algorithm is extended for a system with three antennas at the relay.

\section{Three Antenna Relay} 
\label{Sec:Three Antennas in the Relay}
Assuming three antennas at the relay, it is clear that according to the values of $\lambda_1^{\Sigma}$, $\lambda_2^{\Sigma}$, $\lambda_3^{\Sigma}$, $P_\text{R}$ and $P_\text{S}$, capacity optimal transmission can lead to three different scenarios:
\begin{itemize}
\item\textbf{\textit{Case 1:}} Transmission only via the largest eigenmode achieves capacity (i.e. transmission via the eigenvectors corresponding to $\lambda_1^{\Sigma}$ ); or equivalently, LER is the capacity-optimal transmission method. This scenario will occur only when  \eqref{eq:LER-Constraint} holds. In this case ${\lambda_1^G>0}$ and  ${\lambda_2^G=\lambda_3^G=0}$.
\item\textbf{\textit{Case 2:}} Transmission only via the two largest eigenmodes achieves capacity (i.e. transmission via the eigenvectors corresponding to $\lambda_1^{\Sigma}$ and $\lambda_2^{\Sigma}$); or equivalently, $2$-LER is the capacity optimal transmission method. In this case ${\lambda_1^G>0}$, ${\lambda_2^G>0}$ and ${\lambda_3^G=0}$.
\item\textbf{\textit{Case 3:}} Transmission via all three eigenmodes achieves the capacity; or equivalently, 3-LER is the capacity optimal transmission method. In this case ${\lambda_1^G>0}$, ${\lambda_2^G>0}$ and ${\lambda_3^G>0}$.
\end{itemize}

Note that  \eqref{eq:LER-Constraint} specifies the LER-optimal region. In the following, we intend to specify necessary and sufficient conditions under which, assuming proportional power allocation, transmission only via the two largest eigenmodes in the relay ($2$-LER) approaches capacity.  {Note that proportional power allocation is motivated by the precision of the algorithm introduced in Section~\ref{Sec:Two Antennas in the Relay}. We emphasis that the optimality of $n$-LER is conditioned on proportional power allocation and so it is sub-optimal, however, the results are acceptable when compared with the benchmark\footnote{ {Our conjecture is that the benchmark transmission rate obtained using exhaustive search is, virtually, equivalent to optimal transmission rate.}}.}

\subsection{Conditional Optimality of 2-LER} 
\label{Sec:Optimality of Two Largest Eigenmode Relaying: Necessary Condition}

As discussed in Section~\ref{Sec:Two Antennas in the Relay}, when LER is not optimal, transmission via the two largest eigenmodes, with proportionally assigned power, approximates the benchmark with reasonable accuracy; therefore, by setting
\begin{eqnarray}
\label{eq:L1G=L2G=LeqG}
\lambda_1^G =\lambda_1^{\Sigma}g_2\text{\; and\;} \lambda_2^G = \lambda_2^{\Sigma}g_2
\end{eqnarray}
we aim to derive a necessary and sufficient condition under which, the maximization problem will be achieved by setting ${\lambda_3^G=0}$ and assigning the available power in the relay, proportionally, to  ${\lambda_1^G}$ and ${\lambda_2^G}$.

Considering that $\lambda_1^G\geq \lambda_2^G\geq \cdots \geq\lambda_{n\text{R}}^G$, one can easily conclude that if $\lambda_3^G>0$ leads to rate loss, then all the available power in the relay must be assigned only to ${\lambda_1^G\text{ and } \lambda_2^G }$  and consequently ${\lambda_j^G = 0}$ for ${j\geq 3}$. Let us assume that from the entire available power in the relay, ${\epsilon>0}$ is assigned to $\lambda_3^G$ and, motivated by the results of the two antenna relay scenario in Section \ref{Sec:Two Antennas in the Relay}, the rest of the power is proportionally distributed between $\lambda_1^G$ and $\lambda_2^G$. It is easy to conclude from ${
\frac{\Delta C_{av}(\lambda_3^G)}{\Delta \lambda_3^G} = \frac{C_{av}(\lambda_3^G=\epsilon) - C_{av}(\lambda_3^G=0)}{\epsilon}\leq 0
}$ that ${C_{av}(\lambda_3^G=0) \geq C_{av}(\lambda_3^G=\epsilon)}$; therefore, assigning $\epsilon$ power to $\lambda_3^G$ will cause a rate loss. Now, let us assume that $\epsilon \to 0$, consequently, $\frac{\Delta C_{av}(\lambda_3^G)}{\Delta \lambda_3^G}\leq 0$
 is equivalent to $\frac{\partial C_{av}(\lambda_3^G)}{\partial \lambda_3^G}|_{\lambda_3^G\to 0} \leq 0 $. Therefore, $\frac{\partial C_{av}(\lambda_3^G)}{\partial \lambda_3^G}|_{\lambda_3^G\to 0}\leq 0 $ specifies a region in which assigning power to $\lambda_3^G$~(and consequently $\lambda_j^G\text{ for} j\geq 3$) results in rate loss, and so one must transmit only via $\lambda_1^G\text{ and }\lambda_2^G$ in this region. 

According to the power constraint in \eqref{eq:RelayPowConstraintLambda}, and assuming that ${\lambda_1^G=\lambda_1^{\Sigma}g_2}$ and ${\lambda_2^G=\lambda_2^{\Sigma}g_2}$, one can calculate the power assigned to $\lambda_1^G$ and $\lambda_2^G$ by calculating $g_2$ in \eqref{eq:g2} as 
\begin{eqnarray}
\label{eq:Power-P}
g_2= P_2-\alpha_2 \lambda^G_3
\end{eqnarray}
with 
\begin{eqnarray}
\label{eq:P2}
P_2=\frac{P_\text{R}}{P_\text{S}(\lambda_1^{\Sigma 2} + \lambda_2^{\Sigma 2})+ N_0(\lambda_1^{\Sigma} + \lambda_2^{\Sigma})}
\\
\label{eq:alpha2}
\alpha_2 = \frac{
P_\text{S}\lambda_3^{\Sigma}
 + N_0}{P_\text{S}(\lambda_1^{\Sigma 2} + \lambda_2^{\Sigma 2})+ N_0(\lambda_1^{\Sigma} + \lambda_2^{\Sigma})}
\end{eqnarray}
where the index of $P_2$ and $\alpha_2$ indicates that the $2$-LER condition is being considered. Substituting \eqref{eq:hR-XY} and \eqref{eq:Power-P} in \eqref{eq:ErgodicCapactityLog-gammaD} and considering that ${\log(a/b) = \log(a)-\log(b)}$, we obtain $C(\cdot)$ in \eqref{eq:S1-S2-S3}, at the top of the next page,
where ${W_i =1+\gamma_\text{S} Y \lambda^\Sigma_i}$.
The conditional optimality region of 
$2$-LER corresponds to a region determined by $\frac{\partial C_{av}(\lambda_3^G)}{\partial \lambda_3^G}|_{\lambda_3^G\to 0}\leq 0$, that can be calculated by combining \eqref{eq:S1-S2-S3} and \eqref{eq:ErgodicCapactityDefinition} which is derived in  \eqref{eq:SecondExpectation-2LER-line1} at the top of page
\newcounter{mytempeqncnt2}
\begin{figure*}[!t]
\normalsize
\setcounter{mytempeqncnt2}{\value{equation}}

\begin{eqnarray}
\label{eq:S1-S2-S3}
C(\cdot)=\log \bigg( 
1+(P_2-\alpha_2 \lambda^G_3)
\Big(
\lambda^{\Sigma 2}_1W_1X_1
+
\lambda^{\Sigma 2}_2W_2X_2
\Big)
\bigg)-
 \log \bigg (
1+(P_2-\alpha_2 \lambda^G_3)(\lambda^\Sigma_1 X_1+\lambda^\Sigma_2 X_2) + \lambda^G_3\lambda^\Sigma_3 X_3
\bigg).
\end{eqnarray}
\hrulefill
\begin{eqnarray}
\label{eq:SecondExpectation-2LER-line1}
\frac{\partial C_{av}(\lambda_3^G)}{\partial \lambda_3^G}|_{\lambda_3^G\to 0} &=& \mathbb{E}\Big\{ \frac{\alpha_2/ P_2}{1+P_2 Z_2}\Big\}+
\mathbb{E}\Big\{ \frac{\lambda^\Sigma_3(1+\gamma_\text{S}Y\lambda^\Sigma_3)X_3}{1+P_2 Z_2}\Big\}
-\mathbb{E}\Big\{
\frac{\lambda^\Sigma_3 X_3 - \alpha_2(\lambda^{\Sigma }_1X_1 + \lambda^{\Sigma}_2X_2)}
{1+ P_2(\lambda^{\Sigma }_1X_1 + \lambda^{\Sigma }_2X_2)}
\Big\}\\
\label{eq:SecondExpectation-2LER-line2}
&=&\frac{\alpha_2}{P_2}\mathbb{E}\{ \frac{1}{1+P_2 Z_2}\}+
\lambda^\Sigma_3 \mathbb{E}\{\frac{1+\gamma_\text{S}Y\lambda^\Sigma_3}{1+P_2 Z_2}\}
-(\frac{\alpha_2}{P_2}+\lambda^{\Sigma}_3)\mathcal{D}(\lambda^{\Sigma}_{1,2},P_2)
\end{eqnarray}
\hrulefill
\vspace*{4pt}
\end{figure*}
with
\begin{eqnarray}
\label{eq:Z2}
Z_2=
\lambda^{\Sigma 2}_1W_1X_1
+\lambda^{\Sigma 2}_2W_2X_2.
\end{eqnarray} 
Note that the random variable $X_3$ at the second expectation operation on the right hand side of \eqref{eq:SecondExpectation-2LER-line1} is independent of $Z_2$ and $Y$, hence, $X_3$ can be removed in the first expectation operation because ${\mathbb{E}\{X_3\}=1}$. To the best of our knowledge, the first expectation and second expectations in \eqref{eq:SecondExpectation-2LER-line1} cannot be further simplified; however, \eqref{eq:SecondExpectation-2LER-line1} can be further simplified to \eqref{eq:SecondExpectation-2LER-line2} where
\begin{eqnarray}
\label{eq:D2}
\mathcal{D}(\lambda^{\Sigma}_{1,2},P_2)=\frac{\Gamma(0,\zeta_1)}{P_2(\lambda^{\Sigma }_1 - \lambda^{\Sigma }_2)}\mathrm{e}^{\zeta_1}
+
\frac{\Gamma(0,\zeta_2)}{P_2(\lambda^{\Sigma}_2 - \lambda^{\Sigma}_1)}\mathrm{e}^{\zeta_2}
\end{eqnarray}
with ${\zeta_i = (P_2\lambda^{\Sigma}_i)^{-1}}$. Then, the conditional optimality region of ${2\text{-LER}}$~(i.e.,  ${\frac{\partial C_{av}(\lambda_3^G)}{\partial \lambda_3^G}|_{\lambda_3^G\to 0}\leq 0}$) can be obtained by some algebraic manipulation of \eqref{eq:SecondExpectation-2LER-line2} according to
 \begin{eqnarray}
 \label{eq:2-LERopt-Final}
 \lambda_3^{\Sigma} \leq
\frac{(\alpha_2+P_2\lambda^{\Sigma}_3)\mathcal{D}(\lambda^{\Sigma}_{1,2},P_2) -\alpha_2\mathbb{E}\lbrace\frac{1}{1+P_2 Z_2} \rbrace   }
{ P_2\mathbb{E}\lbrace \frac{1+\lambda^{\Sigma}_3\gamma_\text{S}  Y}{1+P_2 Z_2}\rbrace}
 \end{eqnarray}
 Note that when $\frac{\partial C_{av}(\lambda_3^G)}{\partial \lambda_3^G}|_{\lambda_3^G\to 0} < 0$~(i.e., not including equality), the necessary condition is also sufficient, and so, the strict ``inequality" of \eqref{eq:2-LERopt-Final} specifies a necessary and sufficient condition, under which $2$-LER is the optimal transmission method. However, when equality in \eqref{eq:2-LERopt-Final} holds~(i.e., meaning that $\frac{\partial C_{av}(\lambda_3^G)}{\partial \lambda_3^G}|_{\lambda_3^G\to 0}= 0$), one should make sure that the optimum point is a maximum point; this can be done by showing 
 $\frac{\partial^2 C_{av}(\lambda_3^G)}{\partial (\lambda_3^G)^2}|_{\lambda_3^G\to 0}<0$. In \cite[App.~D]{MoGo:2014TWC}, it is proved that $\frac{\partial^2 C_{av}(\lambda_2^G)}{\partial (\lambda_2^G)^2}|_{\lambda_2^G\to 0}<0$ is always valid. Following the same lines of proof, one can, similarly, prove that $\frac{\partial^2 C_{av}(\lambda_3^G)}{\partial (\lambda_3^G)^2}|_{\lambda_3^G\to 0}<0$ and so, the inequality in  
\eqref{eq:2-LERopt-Final} is a \emph{necessary} and \emph{sufficient condition} under which the ${2\text{-LER}}$ transmission is the optimal\footnote{Optimal in the sense that we assume proportional power allocation to the two largest eigenvectors.  {In the rest of the paper all $n$-LER transmissions are conditioned on proportional power allocation}} transmission method.

\begin{figure}[t]
\begin{center}
        \includegraphics[width=0.45\textwidth, height =0.41 \textwidth]{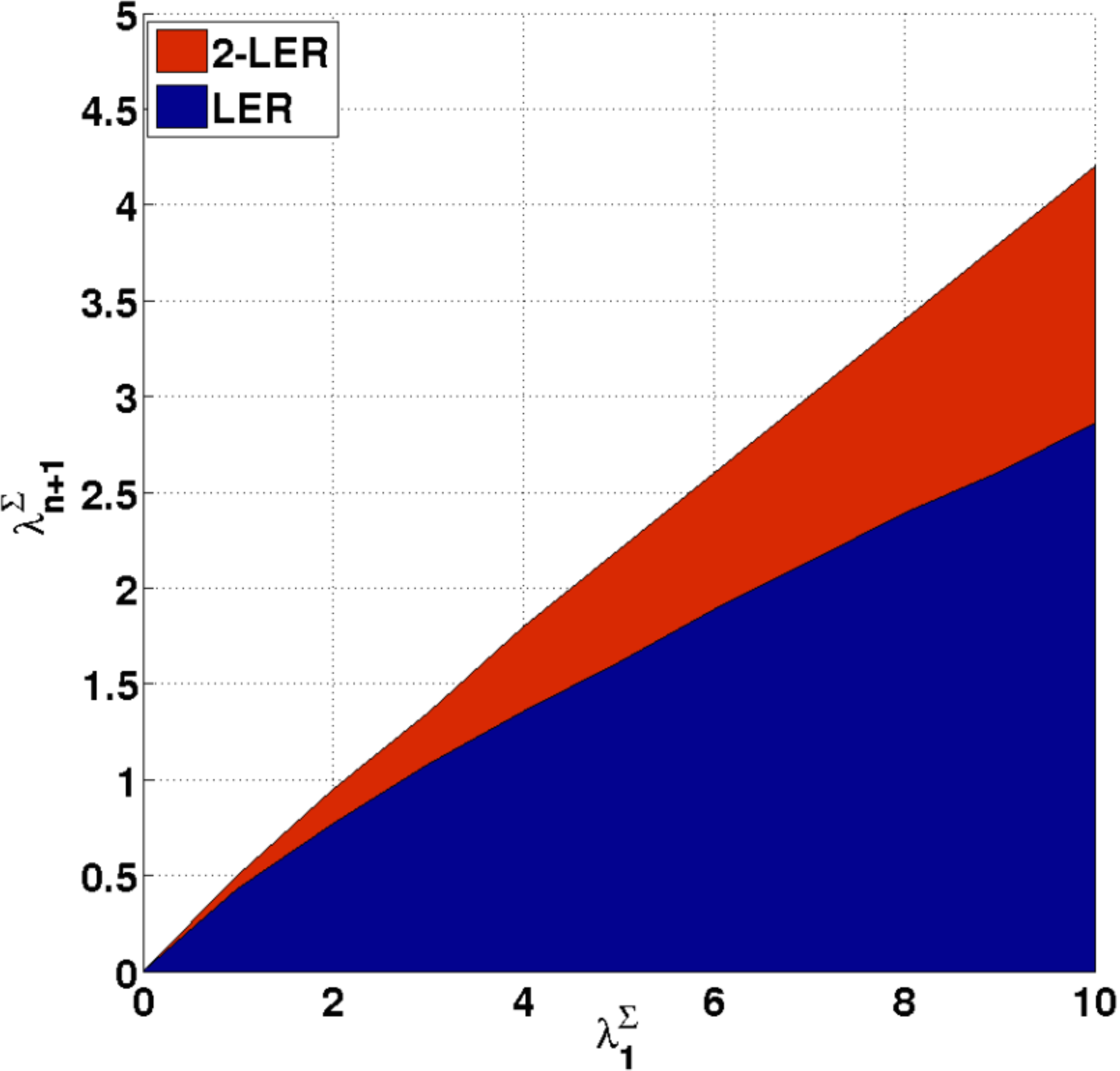}
\end{center} 
\caption{
Blue area: ${\lambda_2^{\Sigma}}$ vs. ${\lambda_1^{\Sigma}}$.  The area shows the LER optimal region.
\hspace*{10mm} Red area: ${\lambda_3^{\Sigma}}$ vs. ${\lambda_1^{\Sigma}}$.  The area illustrates the $2$-LER optimal
\hspace*{10.7mm} region for ${\lambda_2^{\Sigma}=0.5\lambda_1^{\Sigma}}$. For both the areas ${P_\text{R}=P_\text{S}=10\text{ dB}}$.
 }
\label{fig:2-LER}
\end{figure}

\subsection{Near Optimal Power Allocation in a Relay with Three Antennas}
\label{Sec:Near Optimal Power Allocation}
Similar to the algorithm of Section~\ref{Sec:Two Antennas in the Relay}, introduced for power allocation in a two antenna relay, a power allocation method will be introduced for the three cases (\textbf{\textit{Case 1/2/3}}) discussed at the beginning of this section. 

\paragraph*{\textbf{Proposition}} Allocate all available power in the relay to the largest eigenmode \emph{if} (\ref{eq:LER-Constraint}) holds~(i.e., LER), otherwise check (\ref{eq:2-LERopt-Final}) and allocate proportional power to the two largest eigenmodes if  (\ref{eq:2-LERopt-Final}) holds~(i.e., $2$-LER); in the case when neither (\ref{eq:LER-Constraint}) nor (\ref{eq:2-LERopt-Final}) hold, then allocate the power proportionally to all three eigenmodes; i.e., $\lambda_1^G = \lambda_1^\Sigma g_3$, $\lambda_2^G = \lambda_2^\Sigma g_3$ and $\lambda_3^G = \lambda_3^\Sigma g_3$ where
\begin{equation}
\label{eq:g3}
g_3 = \frac{P_\text{R}}{P_\text{S}(\lambda_1^{\Sigma 2}+\lambda_2^{\Sigma 2}+\lambda_3^{\Sigma 2})
+N_0(\lambda_1^\Sigma+\lambda_2^\Sigma+\lambda_3^\Sigma)}
\end{equation}
The algorithm is summarized in Table~\ref{table:nR=3}. $\hfill\blacksquare$

 {Fig.~\ref{fig:C_3RelAnt} illustrates the transmission rates using the proposed algorithm. Comparison with the benchmark confirms that the proposed algorithm is effectively optimal.}

\begin{table}[h]
\caption{Power Allocation in a Relay with $n_\text{R}=3$}
\centering
\begin{tabular}{l}
\hline
\hline
Step 1: Set \:$ P_1=\frac{P_\text{R}}{N_0+P_\text{S}\lambda_1^{\Sigma}}$ \\ 
Step 2: Check the inequality in (\ref{eq:LER-Constraint})\\ 
Step 3: If Step 2 is true \\
\hspace*{15mm} Set $\lambda_1^G =P_1$,  $\lambda_2^G=\lambda_3^G = 0$ and Quit. \\
Step 4: Set $P_2$ from \eqref{eq:P2} and $\alpha_2$  from \eqref{eq:alpha2}\\
Step 5: Check the inequality in (\ref{eq:2-LERopt-Final})\\ 
Step 6: If Step 5 is true \\
\hspace*{15mm} Set 
$\lambda_1^G=\lambda_1^{\Sigma}g_2$, $\lambda_2^G=\lambda_2^{\Sigma}g_2$.,  $\lambda_3^G = 0$ and Quit. \\
\hspace*{10.7mm} else\\
\hspace*{15mm} Set $\lambda_1^G=\lambda_1^{\Sigma}g_3$, $\lambda_2^G=\lambda_2^{\Sigma}g_3$ and $\lambda_3^G=\lambda_3^{\Sigma}g_3$
\vspace*{0.5mm}\\
\hline
\end{tabular}
\label{table:nR=3}
\end{table}

\section{Proposed Power Allocation in a Relay with an Arbitrary Number of Antennas}
\label{Sec:n Antennas in the Relay}
In this section, let us assume that the relay node is equipped with an arbitrary number of the antennas (say $n_\text{R}$). As discussed in earlier sections, in order to achieve capacity, the relay must assign an appropriate amount of its available power per eigenvector.   Following the same approach discussed in Section~\ref{Sec:Three Antennas in the Relay}, one can assume $\kappa$  cases where, depending on the the system parameters (i.e.,~$\boldsymbol{\Sigma}$, $P_\text{S}$ and $P_\text{R}$), $n$-LER will be the capacity approaching transmission method; it means that transmission via $n$ (${n \leq \kappa}$) eigenvectors approaches capacity, and so, only $n$ eigenvectors should be assigned power and the rest of the eigenvectors should be set to zero (i.e., $\lambda_j^G>0 \text{ for } j\leq n$ and $\lambda_j^G=0 \text{ for } j> n$ ).
\begin{figure}[t]
\begin{center}
        \includegraphics[width=0.42\textwidth, height =0.40 \textwidth]{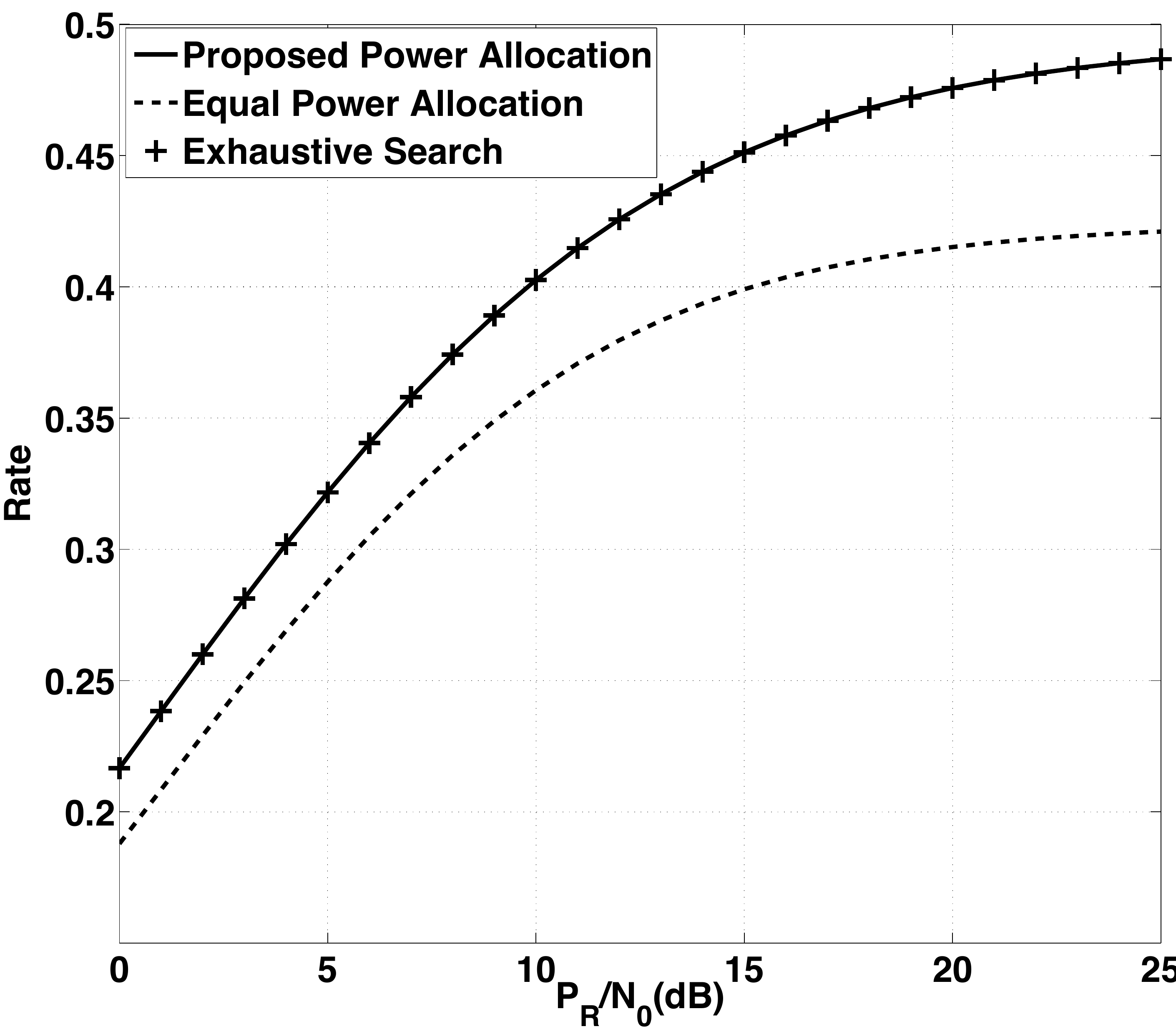}
\end{center} 
\caption{ {Rate vs. transmit power at the relay. The rate values are obtained using the proposed power allocation algorithm, equal power allocation and the benchmark. $n_\text{S} =2$, $n_\text{R}=3$, ${N_0 = 1}$, ${P_\text{S}=0~\text{dB}}$, and inter-antenna correlation ${\rho_{12} = 0.7}$, ${\rho_{23} = 0.5}$ and ${\rho_{13} = 0.2}$. The exhaustive search is performed over discrete values of transmission rates that were obtained using various values of ~${\lambda^G_1\text{ and }\lambda^G_2}$ (step size $0.1$~dB) that fulfil~(\ref{eq:RelayPowConstraintLambda}).} 
 }
\label{fig:C_3RelAnt}
\end{figure}

In Section \ref{Sec:Two Antennas in the Relay}, we introduced the necessary and sufficient condition under which transmission via one eigenmode~(largest eigenmode) achieves capacity; later on, in Section \ref{Sec:Three Antennas in the Relay}, a necessary and sufficient condition was derived, under which, transmission via the  two largest eigenmodes approaches maximum transmission rate. One can extend the same concept and derive a necessary and sufficient condition under which transmission via $n$ (${n \leq \kappa}$) eigenvectors will maximize the rate with the assumption of proportional power allocation.

Analogous to Section \ref{Sec:Three Antennas in the Relay},   let us assume 
that from the available power in the relay, ${\epsilon>0}$ is assigned 
to $\lambda_{n+1}^G$ and the rest of the power is proportionally distributed 
between ${\lambda_1^G\cdots\lambda_n^G}$. It is easy to conclude  
that  $\frac{\partial C_{av}(\lambda_{n+1}^G)}{\partial \lambda_{n+1}^G}|_{\lambda_{n+1}^G\to 0}\leq 0 $ is equivalent to the fact that assigning 
power to ${\lambda_{n+1}^G}$ will result in rate loss, and 
consequently, since 
${\lambda_1^G\geq \lambda_2^G\geq \cdots \lambda_{n_\text{R}}^G }$, the egienvectors corresponding to 
$\lambda_{j}^{\Sigma}$ for ${j\geq n+1}$ should not be assigned power.

 \begin{thm}
The necessary and sufficient condition under which transmission from $n$ largest eigenmodes~($n$-LER transmission) approaches capacity is:
 \begin{eqnarray}
 \label{eq:n-LERopt}
 \lambda_{n+1}^{\Sigma} \leq
\frac{(\alpha_n+P_n\lambda^{\Sigma}_{n+1})\mathcal{D}(\lambda^{\Sigma}_{1,\cdots, n},P_n) -\alpha_n\mathbb{E}\lbrace\frac{1}{1+P_n Z_n} \rbrace   }
{ P_n\mathbb{E}\{ \frac{1+\lambda^{\Sigma}_{n+1}\gamma_\text{S}  Y}{1+P_n Z_n}\}}
 \end{eqnarray} 
where
\begin{eqnarray}
\label{eq:P_n}
P_n=\frac{P_\text{R}}{N_0\sum_{m=1}^{n}\lambda_m^{\Sigma} +P_\text{S}\sum_{m=1}^{n}\lambda_m^{\Sigma 2}}\\
\label{eq:alpha_n}
\alpha_n = \frac{
N_0+P_\text{S}\lambda_{n+1}^{\Sigma}
}{N_0\sum_{m=1}^{n}\lambda_m^{\Sigma} +P_\text{S}\sum_{m=1}^{n}\lambda_m^{\Sigma 2}}
\end{eqnarray}
and
\begin{eqnarray}
\label{eq:Dn}
\mathcal{D}(\lambda^{\Sigma}_{1,\cdots, n},P_n) = 
\sum_{m=1}^{n}\frac{(\lambda^{\Sigma}_m)^{n-2}}
{
\prod\limits_{{k=1}\atop{k\neq m}}^{n} (\lambda^{\Sigma}_m - \lambda^{\Sigma}_k)
}
\Gamma(0,\zeta_m)
\mathrm{e}^{\zeta_m}
\end{eqnarray}
 \end{thm}
 \begin{proof}
 The proof is similar to that of Section \ref{Sec:Optimality of Two Largest Eigenmode Relaying: Necessary Condition}.
 \end{proof}
 Note that \eqref{eq:n-LERopt} is a general form of \eqref{eq:LER-Constraint} and \eqref{eq:2-LERopt-Final}. In fact, \eqref{eq:n-LERopt} will determine the eigenvectors that should be assigned power proportional to their correlation power. The power allocation algorithm is summarized in Table~\ref{table:nR=n}. 
The algorithm is indeed analogous to the water-filling algorithm: the purpose is to find the eigenmodes that should be assigned power for transmission and to discard the ``weak" eigenmodes that result in rate-loss if assigned power. Note that the expressions derived in \eqref{eq:LER-Constraint},
\eqref{eq:2-LERopt-Final} and \eqref{eq:n-LERopt} involve expectation operations that do not seem to lend themselves to calculation in closed-form; therefore, in practical implementation of the system one should implement them using numerical methods. Please see \cite{MoGo:2014TWC} wherein a numerical integration expression is derived.

 Fig.~\ref{fig:C_4RelAnt} illustrates the transmission rates for the proposed algorithm for a relay with four antennas and correlation coefficients as described in the caption. Clearly the proposed algorithm agrees with the rates obtained using exhaustive search. Moreover its superiority over equal power transmission is evident.
The values chosen for inter-antenna correlation in the  numerical simulations in Figs.~\ref{fig:C_3RelAnt} and \ref{fig:C_4RelAnt}  are selected to be relatively high, and to reflect what might be expected in a relatively closely-spaced linear array.  In these cases , the proposed algorithm demonstrates performance  extremely close to the optimum. Note that when $\rho_{ij}\to 0$~(smaller inter-antenna correlation), the spatial correlation diminishes and so the equal power transmission will be the optimum method, which is indeed guaranteed by the proportional power allocation proposed in \eqref{eq:proportional power}; this is demonstrated by numerical simulation in Fig.~\ref{fig:C_2RelAnt}.
Hence, since our algorithm is provably optimum at low correlation and is shown by simulation to be very close to optimum for a typical case of high correlation, it is at least a reasonable hypothesis that it is near optimum for all cases of practical interest.  

\begin{table}[h]
\caption{Power Allocation in a Relay with arbitrary $n_\text{R}$}
\centering
\begin{tabular}{l}
\hline
\hline
Initiate $n=1$\\
\textit{while} $n\leq\kappa$\\
step 1: Set $P_n$ from \eqref{eq:P_n} and $\alpha_n$ from \eqref{eq:alpha_n}\\
step 2: Check the inequality in \eqref{eq:n-LERopt}\\ 
step 3: If Step 2 is true \\
\hspace*{7mm} Set $\lambda_j^G=\lambda_j^{\Sigma}g_n$ for $j\leq n$ \\
\hspace*{7mm} Set  $\lambda_j^G = 0$ for $j> n+1$ and Quit \textit{while}.\\
 else\\
\hspace*{7mm} Set $n \leftarrow n+1$ and go to Step~1\\
end (end of \textit{while} when $n>\kappa$)
\vspace*{0.5mm}\\
\hline
\end{tabular}
\label{table:nR=n}
\end{table}  
 \begin{figure}[t]
\begin{center}
        \includegraphics[width=0.43\textwidth, height =0.40 \textwidth]{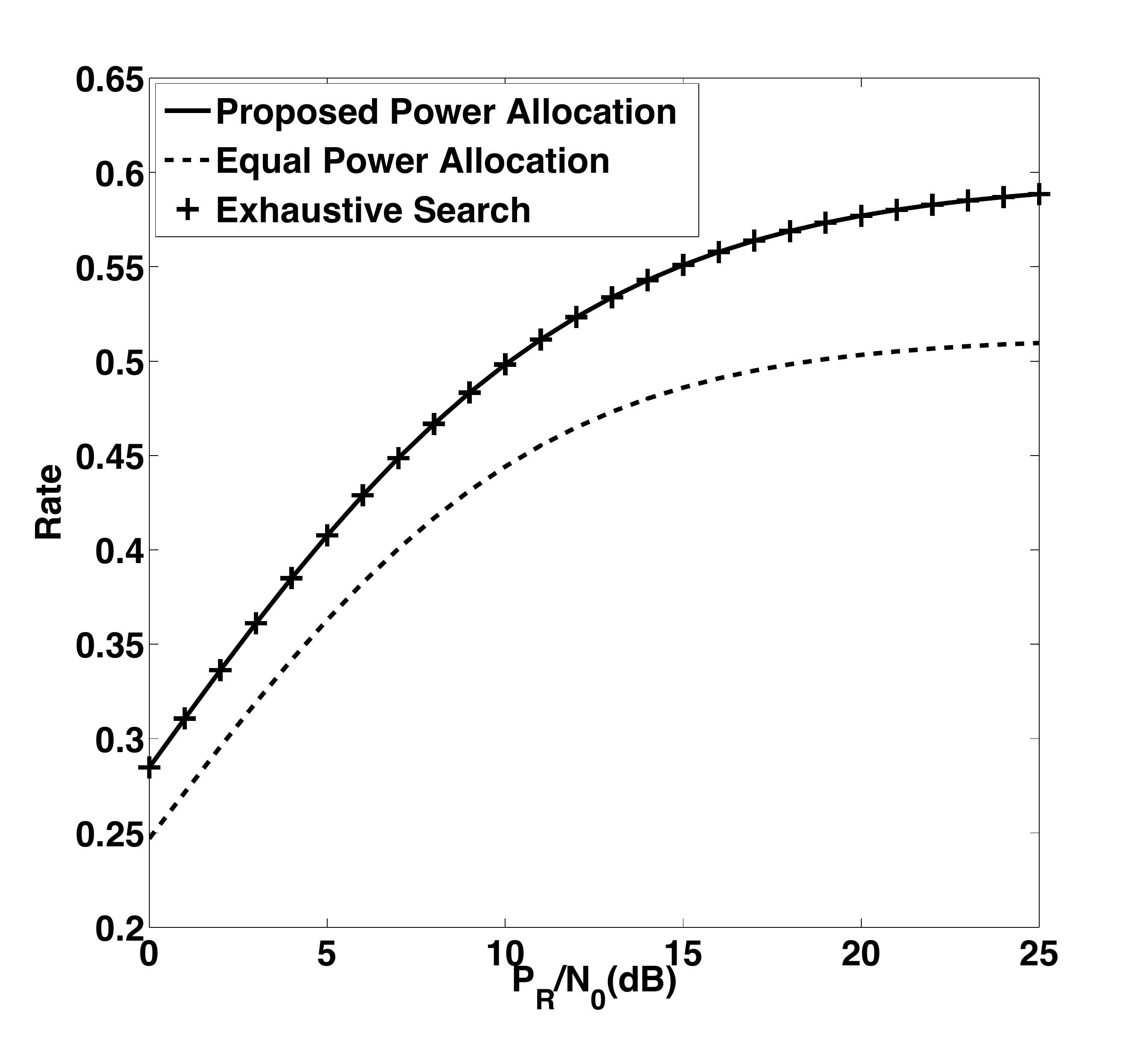}
\end{center} 
\caption{  {Rate vs. transmit power at the relay. The rate values are obtained using the proposed power allocation algorithm, equal power allocation and the benchmark. $n_\text{S} =2$, $n_\text{R}=4$, ${N_0 = 1}$,    ${P_\text{S}=0~\text{dB}}$, and inter-antenna correlation ${\rho_{12} = 0.7}$, ${\rho_{13} = 0.5}$, ${\rho_{14} = 0.3}$, ${\rho_{23} = 0.7}$, ${\rho_{24} = 0.5}$ and ${\rho_{34} = 0.7}$. The exhaustive search is performed over discrete values of transmission rates that were obtained using various values of ~${\lambda^G_i}$ (step size $0.1$~dB) that fulfil~(\ref{eq:RelayPowConstraintLambda}).} 
 }
\label{fig:C_4RelAnt}
\end{figure}

\section{Conclusion}
\label{Sec:Conclusion} 
 This paper studies the statistical characteristics of the received SNR at the destination in a MIMO relay network when the relay node experiences fading correlation. It is assumed that the relay node has access only to the statistical CSI. In order to approach the ergodic capacity of the system, based on the available statistical channel knowledge at the relay, a new relay precoder design methodology is introduced. The proposed method is analogous to the water-filling algorithm; it searches for the largest eigenmodes that should be assigned power and discards the remaining eigenmodes. The simulations demonstrate good agreement between the proposed method and the benchmark which is obtained using exhaustive search.

\appendices
\section{Proof of the SNR Distribution: Full Correlation}
\label{app:SNR Distribution: Full Correlation Scenario}
By combining \eqref{eq:hR-XY} and \eqref{eq:X-FC}, one can write ${\gamma_\text{D} = \gamma_\text{S}n_\text{R}\frac{YV}{1+V}}$, and so, assuming $w=\frac{x}{\gamma_\text{S}n\text{R}}$ we have
\begin{eqnarray}
\label{eq:F-qD-FC}
F_{\gamma_{\text{D}-{FC}}}(x) &=& \mathbb{P}(\frac{YV}
{1+V}<w)=\mathbb{P}(V<\frac{w}{Y-w})\nonumber\\
&&\hspace*{-20mm}1-\mathrm{e}^{\frac{-w}{\lambda_1^G n_\text{R}(Y-
w)}}=\int_0^{\infty}(1-\mathrm{e}^{\frac{-w}
{\lambda_1^G n_\text{R}(y-w)}})f_Y(y)\mathrm{d}y\nonumber\\
&&1 - \int_0^{\infty}\mathrm{e}^{\frac{-w}
{\lambda_1^G n_\text{R}(y-w)}}f_Y(y)\mathrm{d}y
\end{eqnarray}
Note that V is exponentially distributed, and so, ${\mathbb{P}(V<\frac{w}{Y-w})}$ is non-zero only for $Y>0$; consequently, by substituting $f_Y(y)$ in \eqref{eq:F-qD-FC} and  assuming the proper domain for the integral, we have 
\begin{eqnarray}
\label{eq:F-qD-FC}
F_{\gamma_{\text{D}-{FC}}}(x)=1 - \frac{n_\text{S}^{n_\text{S}}}{(n_\text{S}-1)!} \int\limits_w^{\infty}y^{n_\text{S}-1}\mathrm{e}^{-\frac{w}
{\lambda_1^G n_\text{R}(y-w)}-n_\text{S}y} \mathrm{d}y
\end{eqnarray}

by changing variable according to $t=y-w$, assuming ${(t+w)^{n_\text{S}-1}=\sum\limits_{m=0}^{n_\text{S}-1}\frac{(n_\text{S}-1)!}{m! (n_\text{S}-m-1)!}t^m w^{n_\text{S}-m-1} }$ and applying \cite[3.471.9]{GrRy:2007}, $F_{\gamma_{\text{D}-{FC}}}(x)$ will be derived as 
\begin{eqnarray}
\label{eq:CDF-gammaD-FC22}
F_{\gamma_{\text{D}-{FC}}}(w)&=& 1-2(n_\text{S}w)^{n_\text{S}}\mathrm{e}^{-n_\text{S}w}\\
&& \hspace*{-10mm}\times\sum\limits_{m=0}^{n_\text{S}-1}
\frac{(\lambda_1^G n_\text{S}n_\text{R}w)^{-(m+1)/2}}{m! (n_\text{S}-m-1)!}\mathrm{K}_{m+1}(2\sqrt{\frac{n_\text{S}w}{\lambda_1^G n_\text{R}}})
\nonumber
\end{eqnarray}
and so, \eqref{eq:CDF-gammaD-FC} is proved.

\section{Statistics of Random Variable $X$}
\label{app:cdfX}
As a complete proof of \eqref{eq:cdf-X} is lengthy, we only prove the case of $\kappa =2$. Following the same approach, the extension of the proof  to larger values of $\kappa$~(i.e., ${\kappa =3, 4, \cdots, n}$) is straightforward. Then by the rule of mathematical induction, it is easy to obtain the general expression in \eqref{eq:cdf-X}.  

\subsection*{$F_X(x)$ for the Case of $\kappa=2$:}
The numerator and the denominator random variable $X$ in \eqref{eq:X} include summation of random variables $X_j$ which follow exponential distribution with unit mean. For simplicity, we assume $\kappa=2$ (i.e., ${X= \frac{ \lambda_1^G\lambda_1^{\Sigma 2}X_1 + \lambda_2^G\lambda_2^{\Sigma 2}  X_2 }{1+ \lambda_1^G\lambda_1^\Sigma X_1 \lambda_2^G\lambda_2^\Sigma X_2 }}$)
and derive $F_X(x)$ of \eqref{eq:cdf-X}. However, in order to simplify notation, in this section, let us substitute ${\lambda_j^G X_j \to X_j}$ consequently $X_j$ is distributed exponentially with mean $\lambda_j^G$. 
One can write
\begin{eqnarray}
\label{eq:X-proof-start}
F_X(x)=\mathbb{P}(X<x)=\mathbb{P}\big(\frac{ \lambda_1^{\Sigma 2}X_1 + \lambda_2^{\Sigma 2}  X_2 }{1+ \lambda_1^\Sigma X_1 +\lambda_2^\Sigma X_2 } <x
\big).
\end{eqnarray}
 Applying basic algebraic manipulation, \eqref{eq:X-proof-start} can be simplified according to 
\begin{eqnarray}
\label{eq:X-proof-manipulate1}
\mathbb{P}(X<x)&\triangleq &\mathbb{P}\big(
X_1<\frac{
x+\lambda_2^\Sigma(x- \lambda_2^\Sigma)X_2
}{\lambda_1^\Sigma(\lambda_1^\Sigma - x)}
\big)
\end{eqnarray}
Let us split the problem of  deriving $\mathbb{P}(X<x)$ in \eqref{eq:X-proof-manipulate1} over four different intervals: \textit{\textbf{A)}} $x<0$, 
\textit{\textbf{B})} $0<x<\lambda_2^\Sigma$ 
 \textit{\textbf{C)}} $\lambda_2^\Sigma<x<\lambda_1^\Sigma$
  and \textit{\textbf{D)}} $\lambda_1^\Sigma<x$ and derive  $\mathbb{P}(X<x)$ for each interval individually.

 \subsection*{\textbf{A)}  $x<0$:}
  Note that the random variables $X$, $X_1$ and $X_2$ are non-negative random variables. Therefore, for $x<0$, we have ${\mathbb{P}(X<x) =0}$.
 
  \subsection*{\textbf{B)} $0<x<\lambda_2^\Sigma$:}
  Assuming $0<x<\lambda_2^\Sigma$; clearly, $(\lambda_1^\Sigma - x)$ at the denominator of \eqref{eq:X-proof-manipulate1} is positive. Therefore, since $X_1>0$, the expression ${x+\lambda_2^\Sigma(x- \lambda_2^\Sigma)X_2}$ at the numerator of  \eqref{eq:X-proof-manipulate1} must be positive as well and, hence, 
\begin{eqnarray}
\label{eq:X2-positive}
X_2<\frac{-x}{\lambda_2^\Sigma(x-\lambda_2^\Sigma)}
\end{eqnarray}
must hold. Therefore, \eqref{eq:X-proof-manipulate1} can be simplified to 
\begin{eqnarray}
\label{eq:X2-integral}
  \mathbb{P}(X<x) &=&
  \\
  &&\hspace*{-15mm} \int\limits_0^{\frac{-x}{\lambda_2^\Sigma(x-\lambda_2^\Sigma)}}
\mathbb{P}\big(
X_1<\frac{
x+\lambda_2^\Sigma(x- \lambda_2^\Sigma)x_2
}{\lambda_1^\Sigma(\lambda_1^\Sigma -x)}
\big) \underbrace{\frac{\mathrm{e}^{-x_2/\lambda_2^G}}{\lambda_2^G}}_{f_{X_2}(x_2)}
  \mathrm{d}x_2.\nonumber
\end{eqnarray} 
Considering that ${\mathbb{P}(X_1<t) =1-\mathrm{e}^{-t/\lambda_2^G}}$, it is straightforward to prove, from \eqref{eq:X2-integral} , that for $0<x<\lambda_2^\Sigma$,
\begin{eqnarray}
\label{eq:X2-integral-final1}
  \mathbb{P}(X<x) = 1-\frac{c_2}{c_2-c_1}\mathrm{e}^{c_1 x} - \frac{c_1}{c_1-c_2}\mathrm{e}^{c_2 x}.
\end{eqnarray} 
where $c_j$ is defined in \eqref{eq:c-j}.

 \subsection*{\textbf{C)} $\lambda_2^\Sigma<x<\lambda_1^\Sigma$:}
 
 It is easy to deduce from \eqref{eq:X-proof-manipulate1} that for $\lambda_2^\Sigma<x<\lambda_1^\Sigma$, the numerator and the denominator of \eqref{eq:X-proof-manipulate1} is positive for every value of the random variable $X_2$ and, hence, 
\begin{eqnarray}
\label{eq:X2-interval2}
  \mathbb{P}(X<x) = \int\limits_0^{\infty}
\mathbb{P}\big(
X_1<\frac{
x+\lambda_2^\Sigma(x- \lambda_2^\Sigma)x_2
}{\lambda_1^\Sigma(\lambda_1^\Sigma -x)}
\big)\frac{\mathrm{e}^{-x_2/\lambda_2^G}}{\lambda_2^G}
  \mathrm{d}x_2
\end{eqnarray} 
which leads to 
\begin{eqnarray}
\label{eq:X2-integral-final2}
  \mathbb{P}(X<x) = 1-\frac{c_2}{c_2-c_1}\mathrm{e}^{c_1 x} 
\end{eqnarray} 

 \subsection*{\textbf{D)} $x>\lambda_1^{\Sigma}$:}
For the given interval, clearly, the denominator of \eqref{eq:X-proof-manipulate1} is negative while the numerator is positive; hence, one can easily conclude that assuming  $x>\lambda_1^{\Sigma}$ is equivalent to assuming $X_1<0$. On the other hand, since $X_1$ follows the exponential distribution, it is necessarily positive, and so, $x>\lambda_1^{\Sigma}$ is an invalid assumption. Consequently,  $ \mathbb{P}(X<x) = 1 $ for ${x>\lambda_1^{\Sigma}}$.

 \section*{Acknowledgement} 
The work in this paper was supported by the European Commission through the 7th framework programme FP$7$-ICT DIWINE project under contract no.~318177. 

\balance


\begin{thebibliography}{10}
\bibitem{Va:1971}
E.~C. van~der Meulen, ``Three-terminal communication channels,'' {\em Advances
  in Applied Probability}, vol.~3, p.~121, 1971.

\bibitem{LaTsWo:2004}
J.~Laneman, D.~Tse, and G.~Wornell, ``Cooperative diversity in wireless
  networks: Efficient protocols and outage behavior,'' {\em IEEE Transactions
  on Information Theory}, vol.~50, pp.~3062--3080, Dec. 2004.

\bibitem{SeErAa-Part1:2003}
A.~Sendonaris, E.~Erkip, and B.~Aazhang, ``User cooperation diversity. {P}art
  {I}. {S}ystem description,'' {\em IEEE Transactions on Communications},
  vol.~51, pp.~1927--1938, Nov. 2003.

\bibitem{SeErAa-Part2:2003}
A.~Sendonaris, E.~Erkip, and B.~Aazhang, ``User cooperation diversity. {P}art
  {II}. {I}mplementation aspects and performance analysis,'' {\em IEEE
  Transactions on Communications}, vol.~51, pp.~1939--1948, Nov. 2003.

\bibitem{CoGa:1979}
T.~Cover and A.~Gamal, ``Capacity theorems for the relay channel,'' {\em IEEE
  Transactions on Information Theory}, vol.~25, pp.~572--584, Sept. 1979.

\bibitem{Te:1995}
I.~E. Telatar, ``Capacity of multi-antenna {G}aussian channels,'' {\em European
  Transactions on Telecommunications}, vol.~10, pp.~585--595, Nov./Dec. 1999.

\bibitem{GuHa:2008}
W.~Guan and H.~Luo, ``Joint {MMSE} transceiver design in non-regenerative
  {MIMO} relay systems,'' {\em IEEE Communications Letters}, vol.~12,
  pp.~517--519, July 2008.

\bibitem{MoCh:2009}
R.~Mo and Y.~Chew, ``{MMSE}-based joint source and relay precoding design for
  amplify-and-forward {MIMO} relay networks,'' {\em IEEE Transactions on
  Wireless Communications}, vol.~8, pp.~4668--4676, Sept. 2009.

\bibitem{KhRo:2010}
M.~Khandaker and Y.~Rong, ``Joint source and relay optimization for multiuser
  {MIMO} relay communication systems,'' in {\em Proceedings International
  Conference on Signal Processing and Communication Systems (ICSPCS)},
  pp.~1--6, Dec. 2010.
  
  \bibitem{cumanan2013mmse}
K.~Cumanan, Y.~Rahulamathavan, S.~Lambotharan, and Z.~Ding, ``{MMSE}-based
  beamforming techniques for relay broadcast channels,'' {\em IEEE Transactions on Vehicular
  Technology}, vol.~62, no.~8, pp.~4045--4051, 2013.

\bibitem{MoGo:2014TWC}
M.~Molu and N.~Goertz, ``Optimal precoding in the relay and the optimality of
  largest eigenmode relaying with statistical channel state information,'' {\em
  IEEE Transactions on Wireless Communications}, vol.~13, pp.~2113 -- 2123, Apr
  2014.

\bibitem{MuViAg:2007}
O.~Munoz-Medina, J.~Vidal, and A.~Agustin, ``Linear transceiver design in
  nonregenerative relays with channel state information,'' {\em IEEE
  Transactions on Signal Processing}, vol.~55, pp.~2593--2604, June 2007.

\bibitem{TaHu:2007}
X.~Tang and Y.~Hua, ``Optimal design of non-regenerative {MIMO} wireless
  relays,'' {\em IEEE Transactions on Wireless Communications}, vol.~6,
  pp.~1398--1407, Apr. 2007.

\bibitem{DhMcMaBe:2011}
P.~Dharmawansa, M.~McKay, R.~Mallik, and K.~Ben~Letaief, ``Ergodic capacity and
  beamforming optimality for multi-antenna relaying with statistical {CSI},''
  {\em IEEE Transactions on Communications}, vol.~59, pp.~2119--2131, Aug.
  2011.

\bibitem{JeSeLeKiKi:2012}
C.~Jeong, B.~Seo, S.~R. Lee, H.-M. Kim, and I.-M. Kim, ``Relay precoding for
  non-regenerative {MIMO} relay systems with partial {CSI} feedback,'' {\em
  IEEE Transactions on Wireless Communications}, vol.~11, pp.~1698--1711, May
  2012.

\bibitem{FaBe:2008TWC}
G.~Farhadi and N.~Beaulieu, ``On the performance of {A}mplify-and-{F}orward
  cooperative systems with fixed gain relays,'' {\em IEEE Transactions on
  Wireless Communications}, vol.~7, no.~5, pp.~1851--1856, 2008.

\bibitem{Shiu:1999-KAP}
D.~S. Shiu, {\em Wireless Communication Using Dual Antenna Arrays}.
\newblock Kluwer Academic Publishers, Nov. 1999.

\bibitem{Shiu:2000-TC}
D.-S. Shiu, G.~Foschini, M.~Gans, and J.~Kahn, ``Fading correlation and its
  effect on the capacity of multielement antenna systems,'' {\em IEEE
  Transactions on Communications}, vol.~48, pp.~502--513, Mar. 2000.

\bibitem{Zappone:2014-TSP}
A.~Zappone, P.~Cao, and E.~Jorswieck, ``Energy efficiency optimization in
  relay-assisted {MIMO} systems with perfect and statistical {CSI},'' {\em IEEE
  Transactions on Signal Processing}, vol.~62, pp.~443--457, Jan 2014.

\bibitem{ShFoGaKa:2000}
D.-S. Shiu, G.~Foschini, M.~Gans, and J.~Kahn, ``Fading correlation and its
  effect on the capacity of multielement antenna systems,'' {\em IEEE
  Transactions on Communications}, vol.~48, pp.~502--513, Mar. 2000.

\bibitem{Sh:1999}
D.~S. Shiu, {\em Wireless Communication Using Dual Antenna Arrays}.
\newblock Kluwer Academic Publishers, Nov. 1999.

\bibitem{HaAl:2003ICASSP}
M.~O. Hasna and M.-S. Alouini, ``A performance study of dual-hop transmissions
  with fixed gain relays,'' in {\em Proceedings IEEE International Conference
  on Acoustics, Speech, and Signal Processing~(ICASSP '03)}, 2003.

\bibitem{GrRy:2007}
I.~Gradshteyn and I.~Ryzhik, {\em Table of Integrals, Series, and Products}.
\newblock Academic, 7th~ed., 2007.

\bibitem{PaPi:2002}
A.~Papoulis and S.~Pillai, {\em {P}robability, {R}andom {V}ariables and
  {S}tochastic {P}rocesses}.
\newblock McGraw Hill, 4th~ed., 2002.

\end{thebibliography}
\end{document}